\documentclass[letterpaper,11pt]{article}
\usepackage{amsthm}

\newcommand\relatedversion{}

\usepackage[T1]{fontenc}
\usepackage{amsfonts}
\usepackage{graphicx}
\usepackage[inline, shortlabels]{enumitem}

\usepackage[english]{babel}
\usepackage{csquotes}
\usepackage{latexsym,amsmath,amssymb}
\usepackage{colortbl}
\usepackage[dvipsnames]{xcolor}
\usepackage{xspace,bm,setspace,calc}

\usepackage{thmtools}
\usepackage{thm-restate}

\declaretheorem[name=Theorem,numberwithin=section]{theorem}

\newtheorem{obs}{Observation}[section]

\newtheorem{corollary}{Corollary}[section]
\newtheorem{lemma}{Lemma}[section]
\newtheorem{definition}{Definition}[section]

\newtheorem{claim}{Claim}[section]

\graphicspath{{./}{Figures/}}

\setlist[enumerate, 1]{label = (\alph*)}
\setlist{itemsep=0pt}

\usepackage{braket}
\usepackage{mathtools}

\DeclarePairedDelimiter{\ceil}{\lceil}{\rceil}
\DeclarePairedDelimiter{\floor}{\lfloor}{\rfloor}
\DeclarePairedDelimiter{\norm}{\lVert}{\rVert}
\DeclarePairedDelimiter{\intcc}{[}{]}
\DeclarePairedDelimiter{\intco}{[}{)}
\DeclarePairedDelimiter{\intoc}{(}{]}
\DeclarePairedDelimiter{\intoo}{(}{)}

\AddToHook{cmd/appendix/before}{\crefalias{section}{appendix}}

\renewcommand{\Pr}[1]{\ensuremath{\mathbf{Pr} \left[\,#1\,\right]}}
\newcommand{\Exp}{\ensuremath{\mathbb{E}}}

\newcommand{\N}{\ensuremath{\mathbb{N}}}
\newcommand{\R}{\ensuremath{\mathbb{R}}}
\newcommand{\Z}{\ensuremath{\mathbb{Z}}}

\newcommand{\cC}{\ensuremath{\mathcal{C}}}
\newcommand{\cE}{\ensuremath{\mathcal{E}}}
\newcommand{\cF}{\ensuremath{\mathcal{F}}}

\newcommand{\whp}{w.h.p.\xspace}

\newcommand{\poly}{\ensuremath{\mathrm{poly}}\xspace}
\newcommand{\polylog}{\ensuremath{\mathrm{polylog}}\xspace}

\newcommand{\modgreedy}{\mbox{\sc ModulatedGreedy}\xspace}

\newcommand*{\disc}{\operatorname{disc}}
\newcommand*{\adisc}{\operatorname{adisc}}
\newcommand*{\ovld}{\operatorname{ovld}}

\usepackage[textsize=tiny]{todonotes}

\newcommand{\Greedy}[1]{\textsc{Greedy[$#1$]}\xspace}
\newcommand{\Iceberg}[1]{\textsc{Iceberg[$#1$]}\xspace}

\newcommand{\mt}{\ensuremath{m(t)}\xspace}
\newcommand{\avg}{\ensuremath{\lceil\frac{\mt}{n}\rceil}\xspace}
\newcommand{\greedyd}[1]{\mbox{{\sc Greedy}[#1]}\xspace}
\newcommand{\agl}[1]{\mbox{{\sc Left}[#1]}\xspace}

\title{Balls and Bins and the Infinite Process with Random Deletions\relatedversion}
\author{
    Petra Berenbrink\thanks{University of Hamburg, Germany (\email{petra.berenbrink@uni-hamburg.de}).}
    \and
    Tom Friedetzky\thanks{Durham University, U.K. (\email{tom.friedetzky@durham.ac.uk}).}
    \and
    Peter Kling\thanks{Darmstadt University of Applied Sciences, Germany (\email{peter.kling@h-da.de}).}
    \and
    Lars Nagel\thanks{Loughborough University, U.K. (\email{L.Nagel@lboro.ac.uk}).}
}
\date{}

\newcommand*{\isfullversion}{\boolean{true}}  

\usepackage[margin=1in]{geometry}

\usepackage[colorlinks=true,citecolor=purple,linkcolor=blue]{hyperref}
\usepackage[capitalize, nameinlink, noabbrev]{cleveref}

\crefname{obs}{Observation}{Observations}

\newcommand{\email}[1]{\href{mailto:#1}{#1}}

\begin{document}
\maketitle

\begin{abstract}
We consider an infinite balls-into-bins process with deletions where in each discrete step $t$ a coin is tossed as to whether, with probability $\beta(t) \in (0,1)$, a new ball is allocated using the \greedyd{2} strategy (which places the ball in the lower loaded of two bins sampled uniformly at random) or, with remaining probability $1-\beta(t)$, a ball is deleted from a non-empty bin chosen uniformly at random.
Let $n$ be the number of bins and $m(t)$ the total load at time $t$.
We are interested in bounding the \emph{discrepancy} $x_{\max}(t) - m(t)/n$ (current maximum load relative to current average) and the \emph{overload} $x_{\max}(t) - m_{\max}(t)/n$ (current maximum load relative to \emph{highest} average observed so far).

We prove that at an arbitrarily chosen time $t$ the total number of balls above the average is $O(n)$ and that the discrepancy is $ O(\log(n))$.
For the discrepancy, we provide a matching lower bound.
Furthermore we prove that at an arbitrarily chosen time $t$ the overload is $\log\log(n)+O(1)$.
For \enquote{good} insertion probability sequences (in which the average load of time intervals with polynomial length increases in expectation) we show that even the discrepancy is bounded by $\log\log(n)+O(1)$.

One of our main analytical tools is a layered induction, as per~\cite{ABKU99}.
Since our model allows for rather more general scenarios than what was previously considered, the formal analysis requires some extra ingredients as well, in particular a detailed potential analysis.
Furthermore, we simplify the setup by applying probabilistic couplings to obtain certain \enquote{recovery} properties, which eliminate much of the need for intricate and careful conditioning elsewhere in the analysis.
\end{abstract}

\section{Introduction.}%
\label{sec:introduction}

Balls-into-bins processes have been formalized in the mathematical community since at least the early 18th century~\cite{DeMoivre}.
Their objective is, e.g., to allocate at random a set of balls into a number of bins such that the maximum number of balls per bin is minimized.
These processes have been very well studied, with many results being largely folklore.
However, even seemingly minor modifications of the protocol or the model tend to result in vastly different and often far more complex analyses as well as, occasionally, very surprising results.
Apart from the undoubted usefulness as probabilistic paradigms, balls-into-bins processes have many applications, often under the bonnet as primitive black-box operations.
Among others, this includes areas like (the obvious) load balancing and allocation of jobs to servers~\cite{loadbalbook}, routing in networks~\cite{10.1145/276698.276790}, hashing~\cite{doi:10.1137/120871626}, dictionary data structures~\cite{10.1007/11523468_14}, design of file systems~\cite{10.1145/2485732.2485741}, peer-to-peer networks~\cite{10.1145/383059.383071}, and many more.

The most simple balls-into-bins process, which allocates $m$ balls independently and uniformly at random to $n$ bins, is very well understood.
For example, if $m = n$ then the maximum number of balls (\emph{load}) in any bin is, \whp\footnote{
    We say an event $\cE$ happens with high probability (\whp) if $\Pr{\cE} \geq 1 - n^{-\Omega(1)}$.
}, $\log(n)/\log\log(n) \cdot (1+o(1))$~\cite{RaabS98}.
In~\cite{ABKU99} the authors introduced the \Greedy{d} process.
Here, balls are allocated sequentially; each ball chooses $d \in \N$ bins uniformly at random and is allocated to a least loaded of them.
\Greedy{1} is identical to the simple process described above.
But if $d \geq 2$, the maximum load is, \whp, reduced to $\log\log(n)/\log(d)+O(1)$~\cite{ABKU99} for $m = n$.
This exponential decrease in the maximum load is often referred to as \textit{the power of two choices}.
\Greedy{2} was later analyzed for $m \geq n$ (the \emph{heavily loaded case}), which led to the surprising result that, in contrast to \Greedy{1}, the difference between the maximum and average load is independent of $m$~\cite{BCSV06} (namely, \whp it is $\log\log(n) / \log(d) + O(1)$).

\paragraph{Balls-into-Bins with Deletions.}
We study \Greedy{2} in the heavily loaded case where balls are not only inserted but also \emph{deleted}.
Specifically, we ask whether allocation via \Greedy{2} can overcome the negative effect of \emph{random} deletions (effectively \enquote{holes} allocated according to \Greedy{1}, which without insertions would cause a rift of depth $\sqrt{\frac{m}{n}\log n}$ for large $m$).

One can differentiate balls-into-bins with deletions depending on whether, each time a ball is inserted, the ball gets $d$ \emph{fresh} choices (\emph{insertion/deletion} model) or whether the ball might be a previously deleted ball that is reinserted with its \emph{original} $d$ choices (\emph{reinsertion/deletion} model).
Recently, the authors of~\cite{DBLP:conf/focs/0001K22} showed that, somewhat surprisingly, even in the insertion/deletion model, an \emph{oblivious adversary}\footnote{
    The adversary knows only the $d$ choices and insertion time of each ball.
    Each time step it can either
    \begin{enumerate*}
    \item insert a new ball with $d$ fresh random bin choices or
    \item delete the $r$-th most recently inserted ball ($r$ chosen freely).
    \end{enumerate*}
    In particular, it is oblivious to the exact load situation, since it cannot see to which of the $d$ choices the allocation strategy (like \Greedy{d}) assigned a given ball.
} can make sure that \Greedy{2} does not overcome the negative effect of deletions: already for $n = 4$ bins, it can cause a maximum load of $m/4 +\Omega(\sqrt{m})$ with probability $\Omega(1)$.
Intuitively, the lower bound uses the fact that if there is ever a bin $i$ that contains far fewer (say $k$) balls than other bins, \Greedy{2} is biased toward $i$ for the next $k$ throws.
After filling the system to $m$ balls with at least $k$ insertions, the adversary exploits the insertion bias to perform $k$ deletions that are now biased towards bin $i$.
Adding replacement balls for the just deleted balls, one can show that enough of these land on bins $\neq i$ to cause the desired overload.

\paragraph{Results in a Nutshell.}
We consider the insertion/deletion model for $n \in \N$ bins with \emph{random} but \emph{dynamic} deletions, starting from an initially empty system.
Each time step $t$, with probability $\beta(t)$ a new ball with fresh $2$ choices is inserted by \Greedy{2}.
Otherwise, with probability $1 - \beta(t)$, a ball is deleted from a bin chosen uniformly at random among all non-empty bins.
(All our results also hold for the protocol that does not delete from a random bin, but a ball chosen uniformly at random among all balls, see \ifthenelse{\isfullversion}{\cref{sec:randomball}}{the full version}.)
Our main result can be seen as complementing this unfortunate lower bound of \cite{DBLP:conf/focs/0001K22} for \Greedy{2} with a positive result:
In a model with \emph{random} deletions (but still \emph{dynamic}, that is, resulting in a fluctuating average load), \Greedy{2}'s allocation overcomes the negative deletion effect as long as insertions dominate deletions ever so slightly in the recent past of a given time step $t$.
This highlights a fundamental difference between our model and the one of \cite{DBLP:conf/focs/0001K22}. Our result suggests that in many practical applications, where deletions (if at all) occur non-adversarially, \Greedy{2} easily overcomes random \Greedy{1}-type deletions.
This holds true even if the system utilization (ratio between insertions and deletions) fluctuates significantly.
See \cref{sec:contribution} for a full and formal description of our results.

\subsection{Our Contribution.}%
\label{sec:contribution}

The system state at any time can be modeled by a load vector $\bm{x} = \intoo{x_1, \ldots, x_n} \in \N_0^n$, where $x_i$ is the load of the $i$-th fullest bin.
We use $m(\bm{x}) \coloneqq \norm{\bm{x}}_1$ to denote the \emph{total load} of $\bm{x}$ and define, for clarity, $x_{\max} \coloneqq x_1$ and $x_{\min} \coloneqq x_n$ as the \emph{maximum} and \emph{minimum load} of $\bm{x}$, respectively.
We say $\bm{x}$ is \emph{perfectly balanced} if $x_{\min}, x_{\max} \in \set{\floor{m(\bm{x})/n}, \ceil{m(\bm{x})/n}}$.
The \emph{discrepancy} $\disc(\bm{x}) \coloneqq x_{\max} - m(\bm{x})/n$ of $\bm{x}$ is the difference between the current maximum and average load.
In some cases, we also use the stronger \emph{absolute discrepancy} $\adisc(\bm{x}) \coloneqq \max\set{x_{\max} - m(\bm{x})/n, m(\bm{x})/n - x_{\min}}$ of $\bm{x}$ (the largest difference of any load to the current average load).
With this, our process can be described as a random sequence $\intoo[\big]{\bm{x(t)}}_{t \in \N_0}$ of load vectors with total load $m(t) \coloneqq m\intoo[\big]{\bm{x(t)}}$, discrepancy $\disc(t) \coloneqq \disc\intoo[\big]{\bm{x(t)}}$, and absolute discrepancy $\adisc(t) \coloneqq \adisc\intoo[\big]{\bm{x(t)}}$ at time $t$.
We use $m_{\max}(t) \coloneqq \max\set{m(0), \ldots, m(t)}$ to denote the \emph{maximum total load} up to time $t$.
The \emph{overload} $\ovld(t) \coloneqq x_{\max}(t) - m_{\max}(t)/n$ at time $t$  is the difference between the current maximum and the highest average load up to time $t$.
The \emph{height} of a ball is its position in the bin modelled as a stack of balls counting from bottom to top.

We assume a dynamic \emph{insertion probability} $\beta(t)$ that may change from step to step to any value in $\intcc{0, 1}$.
In other words, insertions and deletions are governed by an infinite \emph{insertion probability sequence} $\intoo[\big]{\beta(t)}_{t \in \N}$ giving the (independent) probabilities of an insertion/deletion in each time step.
We say value $x \in \intoo{0, 1}$ is bounded away from $0$ (from $1$) if $x \geq 0 + \Omega(1)$ ($x \leq 1 - \Omega(1)$).
While balls-into-bins processes typically aim to bound the discrepancy, this cannot work in general when deletions are involved. For example, the insertion probability sequence could contain a stretch that deletes $m(t)/2$ balls after some time $t$.
Translating a folklore example (also mentioned in~\cite{DBLP:conf/focs/0001K22}) to our setting, deletions can be regarded as random \enquote{holes} allocated via \Greedy{1}. This would cause a discrepancy of order $\tilde\Theta({\sqrt{m(t)/n}})$, a bound trivially achieved by allocations via \Greedy{1}, even holds if the system starts out perfectly balanced.

Thus, in this paper we aim to bound the \emph{overload}.
We characterize which (possibly local) properties an insertion probability sequence requires for \Greedy{2} to achieve a small discrepancy.
Also, in an infinite process it is quite impossible to \textit{always} have a small discrepancy (or even overload).
Thus, our main results (given below) all hold \emph{with high probability} at an arbitrary (also super-exponentially large) point $t$ in time:
\begin{enumerate}
\item\label{res:1st}
    Assume all insertion probabilities up to time $t$ are bounded away from $0$.
    Then the discrepancy at time $t$ is $O(\log(n))$ (\cref{thm:potential_bound}) and the number of balls above average at time $t$ is $O(n)$ (\cref{prop:linear_potential}).
    The bound on the discrepancy is tight if the insertion probabilities are bounded away from 1/2 from above in the previous $O(n \log(n))$ time steps (see \ifthenelse{\isfullversion}{\cref{lower}}{the full version}).

\item\label{res:2nd}
    Assume we are in a well (e.g., perfectly) balanced situation at time $t$.
    Then we can maintain an overload of $\log\log(n) + O(1)$ for $\poly(n)$ steps as long as the insertion probabilities during that time are bounded away from $0$ and $1$.
    If, additionally, the insertion probabilities ensure an expected load increase for polynomially long subintervals (see definition of \emph{$c$-good time intervals} in \cref{sec:layered}), we can maintain a \emph{discrepancy} (instead of only overload) of $\log\log(n) + O(1)$ (see \cref{thm:maintainmaxdiscrepancy} and \cref{lem:maintainvalidity} in \cref{sec:layered} for the full, formal statements).

\item\label{res:3rd}
    Our final result shows that we achieve an overload of $\log\log(n) + O(1)$ at an \emph{arbitrary} time $t \in \N_0$ as long as the insertion probability sequence till that time is bounded away from $0$ and $1$ (\cref{thm:main3}).
    Similar to the previous result, if the insertion probabilities ensure an expected (possibly quite small) load increase for long subintervals during the previous $n^4$ time steps, we even get a \emph{discrepancy} of $\log\log(n) + O(1)$ at time $t$.
\end{enumerate}
To the best of our knowledge, these are the first results for the discrete, heavily-loaded insertion/deletion setting that provide characterizations for when \Greedy{2} achieves strong guarantees on the \emph{discrepancy} instead of only on the overload.

\paragraph{Main Technical Contributions and Proof Idea.}

To prove that the discrepancy is $O(\log(n))$ (see \ref{res:1st} above) we use a potential function argument similar to~\cite{DBLP:journals/rsa/PeresTW15,DBLP:conf/icalp/TalwarW14}.
To reduce the discrepancy to double-logarithmic (see \ref{res:3rd} above) we use the excellent recovery properties of \Greedy{2}:
In~\cite{DBLP:journals/mst/Czumaj00}, Czumaj studied the recovery of \Greedy{2} for a process that alternates between ball insertions and deletions, maintaining the total load invariant.
He showed that the process is able to recover, i.e., reach a steady state distribution, in polynomial time.
However, this does not reveal any properties (e.g., the discrepancy) of the steady state.

In contrast to~\cite{DBLP:journals/mst/Czumaj00} we cannot have a unique steady state as the total load varies over time.
However, we can show via a path coupling argument that when we start two copies of our process in any two configurations of logarithmic discrepancy with the same number of balls -- which we have at any time via \ref{res:1st} -- it causes both configurations to converge \emph{to each other} quickly (in $o(n^4)$ steps).
This can be interpreted as a \enquote{typical} state to which all configurations with the same number of balls converge.
Again, as in~\cite{DBLP:journals/mst/Czumaj00} this does not yet reveal any further information about such a typical state.

This is where \ref{res:2nd} comes into play, which uses a layered induction to show that, starting from a certain class of configurations (with discrepancy at most $\log\log(n)+c$), the process remains in such a well-balanced configuration for a long  time, $\Omega(n^4)$ steps.
That is, we maintain well-balancedness for longer than it takes to reach the typical state, implying that such typical states lie in this class of well-balanced configurations.
Note that this deviates significantly from previous usages of layered induction in this area:
For example, in~\cite{DBLP:conf/icalp/TalwarW14} the authors used layered induction to \emph{reduce the discrepancy} from $O(\log n)$ to $O(\log\log n)$.
Our layered induction is \enquote{merely} used to establish that a typical state (reached quickly due to the good recovery) has a low discrepancy.
Apart from being flexible enough to deal with a rather general deletion model, our approach also allows for stronger probability guarantees (of order $1 - 1/\poly(n)$ instead of $1 - 1/\polylog(n)$ in~\cite{DBLP:conf/icalp/TalwarW14}).

In our layered induction we have to cope with two new problems in comparison with previous work: The total load $m(t)$ changes over time, and we may not assume it to be linear in $n$.
The classical layered induction proof of~\cite{ABKU99} heavily relies on the base case, which by means of a simple pigeon-hole argument states that there are at most $n/c$ bins with $c$ balls.
This holds because in their setting the total load cannot exceed $n$.
In our case, however, the total load is unbounded.
We deal with this first problem by performing our layered induction only on the highest $\log\log(n) + O(1)$ many load levels.
This leads us to the second challenge, since the changing total load $m(t)$ also implies that the (absolute) height of the considered levels changes over time.
To deal with this, we carefully formulate our layered induction relative to the maximum average and use a flexible random walk argument to prove strong bounds on the number of balls above a certain level over the maximum average.
This yields a bound on the \emph{overload} measure, from which we derive the discrepancy bound -- under the assumption that there was no large, recent drop of the total load.
Note that in the case where such a large drop occurred, our lower bound mentioned in~\ref{res:1st} shows that we cannot hope for a good discrepancy.
To the best of our knowledge, this is the first layered induction for the case $m \gg n$ as well as for a process with a changing load.

\subsection{Related Work.}
There is a vast amount of literature on balls-into-bins games in a multitude of settings. For an overview see \cite{DBLP:journals/fttcs/Wieder17} and for results with one random choice per ball see \cite{DBLP:conf/spaa/BecchettiCNPP15, BFKMNW18, DBLP:journals/corr/abs-2203-12400}. After a very brief overview of results with multiple choices we focus on processes with deletions.

\paragraph{Multiple Choice.}
The \Greedy{d} process was introduced and analyzed in \cite{ABKU99}, although a similar result in a different setting had been shown in \cite{DBLP:conf/stoc/KarpLH92}.
In \cite{BCSV06} the authors generalize the results to the heavily loaded case $m\gg n$. \cite{DBLP:conf/icalp/TalwarW14} presents a simpler analysis for the heavily loaded case, albeit showing a slightly worse upper bound on the maximum load.
The author of \cite{V03} considers the always-go-left process \agl{d} where bins are divided into $d$ clusters of $n/d$ bins each. Every ball randomly chooses one bin per clusters and is allocated to a least loaded among the chosen bins; ties are broken in favor of the leftmost cluster.  This results in a maximum load of $\log\log(n)/(d\cdot \phi_d)$ with $\phi_d\le 2$ ($\phi_d \in \intoo{1, 2}$ is the generalized golden ratio), \whp, a surprising improvement over \cite{ABKU99}. For more recent results see \cite{DBLP:conf/soda/LosSS22,DBLP:conf/soda/LosS023,DBLP:journals/siamdm/LosSS24}.
The potential function analysis for Part \ref{res:1st} is inspired by methods introduced in \cite{DBLP:journals/rsa/PeresTW15} where the authors consider a very interesting variant of \Greedy{d}, which they call $(1+\beta)$-process.
Here each ball is allocated using \Greedy{1} with probability $1-\beta$, and using \Greedy{2} with the remaining probability $\beta$.
The authors showed that the difference between maximum and average load is $\Theta(\log(n)/\beta)$ for each $\beta<1$. Hence, as long as the number of allocations with \Greedy{2} are sufficiently large, the load is rebalanced and the \enquote{damage} caused by \Greedy{1} is undone.
In this paper we use the same potential function and generalize the analysis to deletions (see \cref{sec:larstom}).

\paragraph{Processes with Deletion, $m \leq n$.}
In~\cite{ABKU99} the authors analyze the \Greedy{d} allocation strategy in an insertion/deletion model  where, after $n$ initial insertions, random deletions and insertions happen alternately.
They show that for a polynomial number of steps the maximum load is, \whp, $\log\log(n)/log(d) + O(1)$.
The authors of~\cite{CFMMRSU98} provide an alternative proof for the same bound that can handle arbitrary insertion/deletion sequences, as long as the total system load is at most $n$.
Under the same restriction they provide bounds for the reinsertion/deletion model with adversarial deletions, showing that the maximum load is, \whp, $O(\log\log(n))$.

\paragraph{Processes with Deletions, $m > n$.}
The infinite insertion/deletion model with adversarial deletions for the \emph{moderately loaded case} is considered by~\cite{V03} for the \agl{d} process where the maximum system load is always bounded by $h \cdot n$.
He proves that the maximum (bin) load is $\ln\ln(n)/(d \cdot \ln(\phi_d)) + O(h)$ \whp\ (again, $\phi_d \in \intoo{1, 2}$ is the generalized golden ratio).
Note that for $h = \Omega(\log(n))$, this is not better than the bound achieved by the trivial \Greedy{1} allocation.

The authors of~\cite{DBLP:conf/soda/BenderCFKT23} consider  the moderately loaded case but for the reinsertion/deletion model.
They introduce the \Iceberg{d} process, combining \Greedy{1} and \agl{d}.
Slightly simplified, the \Iceberg{d} process allocates a ball via \Greedy{1} if the target bin contains at most $h + \tilde{O}(\sqrt{h})$ balls that were allocated by \Greedy{1}.
Otherwise, either \agl{d} is used (if there are currently fewer than $n/d$ balls allocated by \agl{d}) or the ball is allocated to bin $1$.
This process achieves a maximum load of $h + \log\log(n)/(d\log \phi_d)) + \tilde{O}(\sqrt{h})$.
Note that this process must know $h$ and keep track of the number of balls placed with specific strategies, making it difficult to use in distributed settings.
In contrast, \Greedy{d} places balls solely based on the selected bins' loads.

The focus of~\cite{DBLP:journals/mst/Czumaj00} is the analysis of the \emph{recovery time} (the time it takes a process to move from an arbitrary to a \enquote{typical} state) in the heavily loaded case ($m$ arbitrary).
The author considers a setting starting with $m$ arbitrarily placed balls followed by alternating deletions (either a random ball or from a random bin) and insertions.
The recovery time under \Greedy{d} allocation is proven to be polynomial in $m$ and $n$.
While parts of our analysis (\cref{sec:coupling}) are inspired by this result, note that in our case the system load $m$ fluctuates over time and that we do not only show fast recovery but actually characterize what the \enquote{typical} state looks like.

As mentioned at the beginning of the introduction,~\cite{DBLP:conf/focs/0001K22} give lower bounds for \Greedy{2} in the heavily loaded, adversarial insertion/deletion model.
They also provide a more general lower bound for any oblivious allocation strategy (like \Greedy{2} or \agl{2}) in the reinsertion/deletion model (with the system load capped at $m$): an oblivious adversary can force a maximum load of $m/4 + m^{\Omega(1)}$ for $n=4$ bins within the first $\poly(m)$ steps.
On the positive side, they introduce the allocation strategy \modgreedy for the insertion/deletion model, which achieves an overload of $O(\log(m))$.
To overcome their lower bound for \Greedy{2}, \modgreedy decides probabilistically between two bin candidates, based on their current load.
The algorithm must know global guarantees on the load situation (in particular, the maximum number of balls that will ever be present and the current maximum and minimum load).

\paragraph{Queuing Systems.}
A prominent example of a queuing system is the supermarket model considered in, e.g.,~\cite{VDK96, DBLP:conf/spaa/BerenbrinkCFV00, DBLP:journals/tpds/Mitzenmacher01, LN05, LM05}.
It can be seen as a variant of balls-into-bins process in continuous time.
Here, balls typically arrive in a Poisson stream of rate $\lambda n$ and are allocated according to \Greedy{d}.
Bins process (delete) balls in first-in first-out order, with a processing time that is exponentially distributed with mean $1$.
These systems and their results can often be readily mapped to the balls-into-bins settings yielding results similar to those discussed above.
However, the focus in queuing systems lies typically on properties like the balls' service times (the average time a ball spends in the system), how long it takes to reach an equilibrium state, or how the system behaves in the equilibrium.
Note that the (expected) unit processing time together with $\lambda < 1$ means that we are dealing here with the lightly/moderately loaded case ($m \leq \poly(n)$) in the equilibrium.
Later results (e.g.,~\cite{DBLP:conf/sigmetrics/BramsonLP10, DBLP:journals/sigmetrics/MukherjeeBLW16, BFL}) considered settings where $d$ and/or $\lambda$ might depend on $n$ or where the balls' service time adhere to more general distributions.

\paragraph{Further Related Work in Load Balancing.}
There are also results such as~\cite{DBLP:conf/esa/AdlerBS98, BFKMNW18} that consider \emph{parallel} load balancing settings, where balls are not sequentially allocated but thrown in batches of size $\Theta(n)$ and where, similar to the supermarket model, bins have a processing time (in these cases simply $1$).
One is interested in which batch sizes and allocation strategies yield small waiting times.
\cite{DBLP:conf/esa/AdlerBS98} shows for \Greedy{2} that, if the batches are not too large, the maximum waiting time is, \whp, at most $\ln\ln(n) / \ln(d) + O(1)$.
For a batch size of $\lambda n$, \cite{BFKMNW18} prove a maximum waiting time of $O(\frac{1}{1-\lambda} \cdot \log \frac{n}{1-\lambda})$ for \Greedy{1} and $O(\log \frac{n}{1-\lambda})$ for \Greedy{2}.

Another line of work considers a mixture of load balancing and averaging processes on general graphs (in contrast to the balls-into-bins settings, that can be thought of as load balancing on a clique of $n$ nodes).
Two recent results are~\cite{DBLP:journals/algorithmica/AlistarhNS22, DBLP:conf/icalp/BerenbrinkHHKR23}.
The major difficulty here typically stems from dealing with the given graph's structure, which changes the considered processes and analysis techniques considerably compared to our setting.

\section{Max Load \& Above-average Ball Count via Potential Functions.}%
\label{sec:larstom}
In this section we extend the analysis of \cite{DBLP:journals/rsa/PeresTW15} to allow for ball deletions (with variable deletion probabilities).
The following \cref{thm:potential_bound} uses the potential function $\Gamma_\alpha(\bm{x(t)})$ from \cite{DBLP:journals/rsa/PeresTW15}.  We show that, at an arbitrary point of time, the expected potential is linear in $n$.
The potential is defined as $\Gamma_{\alpha}(\bm{x(t)}) = \Phi_{\alpha}(\bm{x(t)}) + \Psi_{\alpha}(\bm{x(t)})$, where
\begin{equation}
\Phi_{\alpha}(\bm{x(t)}) = \sum_{i=1}^n e^{\alpha \cdot (x_i(t) - m(t)/n)},
\quad\text{and}\quad
\Psi_{\alpha}(\bm{x(t)}) = \sum_{i=1}^n e^{-\alpha \cdot (x_i(t) - m(t)/n)}
.
\end{equation}

The potential $\Phi_{\alpha}(\bm{x(t)})$ quantifies the impact of bins with loads above the average while $\Psi_{\alpha}(\bm{x(t)})$ does the same for the bins below the average.
Although bounds on the maximum load are derived only from $\Phi_{\alpha}(\bm{x(t)})$, both parts are needed to ensure a potential drop in every load situation where the potential is large.
\cref{thm:potential_bound}, whose proof can be found in \ifthenelse{\isfullversion}{\cref{app:potential_bound}}{the full version}, demonstrates that the potential function is robust enough to tolerate deletions if the $\beta(t)$ are lower-bounded by a constant $\beta > 0$.
It states that the absolute discrepancy at any time $t$ is, \whp, at most $O(\log n)$ higher than it was at the start.
In particular, if the process starts with an absolute discrepancy of $O(\log n)$, it will almost certainly be $O(\log n)$ at time $t$.

\begin{restatable}{theorem}{potentialbound}\label{thm:potential_bound}
Consider an insertion probability sequence $\intoo[\big]{\beta(t)}_{t \in \N}$ bounded away from $0$ by a constant $\beta > 0$.
If $\alpha \leq \beta/16$, then for any $t \in \N_0$ we have $\Exp[\Gamma_\alpha(\bm{x(t)})] = O(n) \cdot e^{\alpha \cdot \adisc(0)}$.
Moreover, for any $a \geq 0$ we have
\begin{math}
\Pr{\adisc(t) - \adisc(0) \geq \frac{a + 2}{\alpha} \cdot \ln n}
\leq n^{-a}
\end{math}.
\end{restatable}
In particular, if we start in a perfectly balanced situation with discrepancy $0$ (e.g., empty system), the expected potential at an \emph{arbitrary} time $t$ is only linear.
Equipped with the above, we can prove a strong upper bound on the number of balls that lie above the average at a given time $t$:

\begin{restatable}{theorem}{restateThmTom}\label{prop:linear_potential}
Fix an arbitrary step $t$. Let $\beta \in (0,1]$ be a constant. Assume $\beta(\tau) \geq \beta$ for all $1 \leq \tau \leq t$ and $\alpha\le \beta/16$. For any fixed non-negative $\ell$ and any fixed $\mu \in (0,1)$, there is a $\gamma=\gamma(\mu)$ so that, 
with probability at least $1-n^{-(3+\ell)}$, the number of balls above height $\lceil \frac{m(t)}{n} \rceil+\gamma$ is at most $\mu n$.
\end{restatable}

The proof is based on\ifthenelse{\isfullversion}{ \cref{lem:adaption_sauerwald} in \cref{app:tom},}{} an adaptation of a result in \cite{DBLP:conf/innovations/LosS22, LS21}, which gives us a high-probability guarantee for the potential $\Gamma_\alpha(\bm{x(t)})$ being linear\ifthenelse{\isfullversion}{}{ (see full version)}.
In fact, the purpose of the adaptation of \cite{DBLP:conf/innovations/LosS22, LS21}'s results is two-fold: (i) to incorporate our deletions and their respective probability bounds, and (ii) to improve the high-probability guarantee from the original $1-1/n^3$ (which was sufficient for the purposes of \cite{DBLP:conf/innovations/LosS22, LS21}).
Note that the probability stated in \cref{prop:linear_potential} is inherited from \ifthenelse{\isfullversion}{\cref{lem:adaption_sauerwald}}{the adapted lemma}; apart from that, the theorem's actual proof is an entirely deterministic counting argument.
The proof details may be found in \ifthenelse{\isfullversion}{\cref{app:tom}}{the full version}.

\section{Maintaining a Well-balanced State.}%
\label{sec:layered}
\newcommand*{\nrballslvlgeq}[2]{m_{#1}(#2)}

Before we state this \lcnamecref{sec:layered}'s main result (\cref{thm:maintainmaxdiscrepancy}), we introduce the notion of \emph{$c$-good time intervals} $I$ for a given insertion probability sequence.
Intuitively, such an interval guarantees us that the total system load increases in expectation in any sufficiently long subinterval of $I$.
\begin{definition}[$c$-good interval]\label{def:good}
Given a sequence $\beta(t)$ ($t \in \N$) of insertion probabilities, we call a discrete time interval $I$ \emph{$c$-good} if for each subinterval $(t_1, t_2] \subseteq I$ with $t_2 - t_1 \geq c \cdot n$ we have
\begin{math}
\left(\sum_{t = t_1 + 1}^{t_2}\beta(t)\right)/(t_2 - t_1)
\geq
\frac12 (1 + \epsilon)
\end{math} for a constant $\epsilon > 0$.
\end{definition}

\Cref{thm:maintainmaxdiscrepancy} states that, \whp, our process, when started in a perfectly balanced configuration, remains well-balanced for at least $n^4$ steps.\footnote{
    In fact, our proof shows that starting from \emph{any} well-balanced configuration, the situation remains well-balanced for polynomial many steps.
}
Here, well-balanced (this is formalized later in \cref{def:validetc} and \cref{lem:maintainvalidity}) is measured relative to the maximum average observed so far for general insertion probability sequences and relative to the current average if the considered time window is $O(1)$-good.
\begin{theorem}%
\label{thm:maintainmaxdiscrepancy}
Consider an insertion probability sequence $\intoo[\big]{\beta(t)}_{t \in \N}$ bounded away from $0$ and $1$.
Let $\bm{x(0)}$ be a perfectly balanced configuration.
Then, \whp,
\begin{math}
x_{\max}(t)
\leq
 m_{\max}(t)/n + \log\log n + O(1)
\end{math}
for any $t \leq n^4$.
Moreover, if there is a constant $c \in \N$ such that the time interval $\intcc{0, n^4}$ is $c$-good, then, \whp,
\begin{math}
x_{\max}(t)
\leq
m(t)/n + \log\log n + O(1)
\end{math}
for any $t \leq n^4$.
\end{theorem}
To prove \cref{thm:maintainmaxdiscrepancy}, we conduct a layered induction. A major difficulty stems from the fact that we have to deal with a permanently changing average load.
For this an important ingredient is \cref{prop:linear_potential}, which gives us a high probability bound on the number of balls above the \emph{current} average.
This enables us to derive a base case for our layered induction, independent of the changing average load.
The layered induction works roughly as follows:
For each \emph{level} $\ell \in \set{0, 1, \dots, \log\log(n) + \Theta(1)}$ we consider the number of balls at height greater than average $+\ell$.
For each level we define a \emph{critical threshold} $\alpha_{\ell}$ that decreases doubly exponentially with $\ell$.
Moreover, each level has a \emph{safe threshold} $\alpha_{\ell}/2$.
We show that, as long as all levels remain below their critical threshold, any level that is above its safe threshold has a strong drift back \emph{towards} the safe threshold.
A standard result on random walks then shows that it is unlikely that during the $\poly(n)$ many considered time steps, any of the $\log\log(n) + \Theta(1)$ many levels (that all start below their respective safe threshold) crosses the \emph{critical interval} $\intco{\alpha_{\ell/2}, \alpha_{\ell}}$.

\paragraph{Formalizing Well-balanced Configurations.}
In the following, let $\hat{\beta} \in \intoo{0, 1}$ be a constant upper bound on the insertion probability of our process, such that $\beta(t) \leq \hat{\beta}$ for all considered time steps $t$.
For a height $h \in \N_0$ and a configuration $\bm{x}$ define $m_h(\bm{x})$ as the number of balls that have height at least $h$ in configuration $\bm{x}$.
For our process starting in an initial configuration $\bm{x(0)}$, we use the shorthand $m_h(t) \coloneqq m_h(\bm{x(t)})$.
In our analysis, it will be helpful to measure the number of balls at or above some \emph{level} $\ell$, which, in turn, is above a given \emph{base height} $h$ (which might vary over time).
To this end, we introduce the notation $m_{\ell}^{(h)}(\bm{x}) \coloneqq m_{h + \ell}(\bm{x})$ and $m_{\ell}^{(h)}(t) \coloneqq m_{h + \ell}(t)$.

We consider $\ell^* + 2 = \log\log(n) + \Theta(1)$ many levels $\ell \in \set{0, 1, \dots, \ell^* + 1}$ with their \emph{critical thresholds} $\alpha_{\ell}$ defined as follows:
\begin{equation}%
\label{eqn:lvlboundary:def}
\alpha_{\ell} \coloneqq \begin{cases}
    \frac{1 - \hat{\beta}}{128\hat{\beta}} \cdot n                      & \text{if $\ell = 0$,}
    \\
    \frac{32\hat{\beta}}{1 - \hat{\beta}} \cdot \alpha_{\ell - 1}^2 / n & \text{if $\alpha_{\ell - 1} > \sqrt{\frac{3(1 - \hat{\beta})}{2\hat{\beta}} \cdot n \log(n)}$,}
    \\
    12\log(n)                                                           & \text{if $\ell = \ell^* \coloneqq \min\Set{\ell | \alpha_{\ell - 1} \leq \sqrt{\frac{3(1 - \hat{\beta})}{2\hat{\beta}} \cdot n \log(n)}}$, and}
    \\
    24                                                                  & \text{if $\ell = \ell^* + 1$.}
\end{cases}
\end{equation}
Note that the recursive definition ensures $\ell^* = \log\log(n) + \Theta(1)$.
Moreover, it implies the following relations between the critical thresholds $\alpha_{\ell}$ and $\alpha_{\ell - 1}$:
\begin{obs}%
\label{obs:lvlboundary:props}
For any $\ell \in \set{1, 2, \dots, \ell^* + 1}$ and large enough $n$ we have
\begin{equation}
\frac{8\hat{\beta}}{1 - \hat{\beta}} \cdot \frac{\alpha_{\ell - 1}^2}{n}
\leq
\alpha_{\ell}
\leq
\frac{1}{4} \cdot \alpha_{\ell - 1}
.
\end{equation}
\end{obs}
%

Using the critical threshold $\alpha_{\ell}$, we define the \emph{valid interval} of level $\ell \in \set{0, 1, \dots, \ell^* + 1}$ as $V_{\ell} \coloneqq \intco{0, \alpha_{\ell}}$.
Similarly, we define the level's \emph{safe interval} as $S_{\ell} \coloneqq \intco{0, \alpha_{\ell}/2}$ and its \emph{critical interval} as $C_{\ell} \coloneqq \intco{\alpha_{\ell}/2, \alpha_{\ell}}$.

\begin{definition}\label{def:validetc}
For a base height $h \in \N_0$ and level $\ell \in \set{0, 1, \dots, \ell^* + 1}$ we say level $\ell$ of configuration $\bm{x}$ is
\begin{enumerate*}[afterlabel=, label=]
\item \emph{$h$-valid} if $m_{\ell}^{(h)}(\bm{x}) \in V_{\ell}$,
\item \emph{$h$-safe} if $m_{\ell}^{(h)}(\bm{x}) \in S_{\ell}$,
\item  and \emph{$h$-critical} if $m_{\ell}^{(h)}(\bm{x}) \in C_{\ell}$.
\end{enumerate*}
Configuration $\bm{x}$ is \emph{$h$-valid}/\emph{-safe}/\emph{-critical} if all levels $0$ to $\ell^* + 1$ are $h$-valid/-safe/-critical.
\end{definition}
In proofs, we sometimes omit the base height if it is clear from the context.

\paragraph{Analysis.}
We start with a simple \lcnamecref{lem:muchangeprob} that provides the probabilities for the increase and decrease of the number of balls at or above a given height $h$.
\begin{lemma}%
\label{lem:muchangeprob}
For any configuration $\bm{x}$ and any height $h \in \N$ we have
\begin{enumerate}
\item
    \begin{math}
    \Pr{\nrballslvlgeq{h}{t+1} > \nrballslvlgeq{h}{t} \;|\; \bm{x(t)} = \bm{x}}
    =
    \beta(t) \cdot {\bigl(\nrballslvlgeq{h - 1}{t} - \nrballslvlgeq{h}{t}\bigr)}^2 / n^2
    \end{math}
    and
\item
    \begin{math}
    \Pr{\nrballslvlgeq{h}{t+1} < \nrballslvlgeq{h}{t} \;|\; \bm{x(t)} = \bm{x}}
    \geq
    \bigl( 1 - \beta(t) \bigr) \cdot \bigl(\nrballslvlgeq{h}{t} - \nrballslvlgeq{h + 1}{t}\bigr) / n
    .
    \end{math}
\end{enumerate}
\end{lemma}
\begin{proof}
\begin{enumerate}
\item For $\nrballslvlgeq{h}{t}$ to increase during the next step, a ball must be inserted by \Greedy{2} (probability $\beta(t)$) and both of the choices of \Greedy{2} must fall into a bin that has load at least $h - 1$.
    The number of such bins is exactly $\nrballslvlgeq{h-1}{t} - \nrballslvlgeq{h}{t}$, yielding the probability ${\bigl(\nrballslvlgeq{h - 1}{t} - \nrballslvlgeq{h}{t}\bigr)}^2 / n^2$.

\item For $\nrballslvlgeq{h}{t}$ to decrease during the next step, a ball must be deleted (probability $1 - \beta(t)$), and the chosen bin must have load at least $h$.
    The number of such bins is exactly $\nrballslvlgeq{h}{t} - \nrballslvlgeq{h+1}{t}$ and the total number of non-empty bins is at most $n$, yielding the probability $\geq \bigl( 1 - \beta(t) \bigr) \cdot \bigl(\nrballslvlgeq{h}{t} - \nrballslvlgeq{h + 1}{t}\bigr) / n$.
    \qedhere
\end{enumerate}
\end{proof}

\begin{restatable}{lemma}{restatemurwdrift}%
\label{lem:murw_drift}
Consider an $h$-valid configuration $\bm{x}$ and let $\bm{x'}$ denote the (random) configuration after one step of our process with the insertion probability $\beta(t) \leq \hat{\beta} < 1$.
\begin{enumerate}
\item If level $\ell \in \set{1, 2, \dots, \ell^*}$ of configuration $\bm{x}$ is $h$-critical, then
    \begin{equation}
    \Pr{m_{\ell}^{(h)}(\bm{x'}) < m_{\ell}^{(h)}(\bm{x})} \Big/ \Pr{m_{\ell}^{(h)}(\bm{x'}) > m_{\ell}^{(h)}(\bm{x})}
    \geq
    2
    .
    \end{equation}

\item If level $\ell = \ell^* + 1$ of configuration $\bm{x}$ is $h$-critical and if $n$ is large enough, then
    \begin{equation}
    \Pr{m_{\ell}^{(h)}(\bm{x'}) < m_{\ell}^{(h)}(\bm{x})} \Big/ \Pr{m_{\ell}^{(h)}(\bm{x'}) > m_{\ell}^{(h)}(\bm{x})}
    \geq
    \sqrt{n}
    .
    \end{equation}
\end{enumerate}
\end{restatable}
The proof of \cref{lem:murw_drift} may be found in \ifthenelse{\isfullversion}{\cref{app:layered}}{the full version}.

Next we formulate an auxiliary \lcnamecref{lem:biasedrw_crossprob} about the probability that a random walk on integers reaches position $b > 0$ before position $-a < 0$ when started in $0$, assuming the random walk is biased towards the left by a factor $r > 0$.
We will apply this \lcnamecref{lem:biasedrw_crossprob} to the values $m_{\ell}^{(h)}(\bullet)$ (basically corresponding to the random walks) to show that these values are unlikely to cross the critical interval $C_{\ell}$ (on which, by \cref{lem:murw_drift}, $m_{\ell}^{(h)}(\bullet)$ is biased by a factor of at least $2$ towards the safe left side).
\begin{lemma}%
\label{lem:biasedrw_crossprob}
Consider a (not necessarily i.i.d.) random walk $(S_t)_{t \in \N_0}$ on $\Z$ with initial position $S_0 = 0$ and position $S_t = \sum_{i = 1}^{t} X_i$ after $t$ steps with \emph{step width} $X_i \in \set{-1, 0, +1}$ at time $i$.
Let $\cF = \intoo{\cF_t}_{t \in \N_0}$ denote the random walk's natural filtration.
Assume there is $r \in \R_{>0} \setminus \set{1}$ such that for all $t \in \N$,
$\Pr{X_t = -1 \;|\; \cF_{t-1}} ~\big/~ \Pr{X_t = +1 \;|\; \cF_{t-1}} = r$.
For $x \in \Z$ define the stopping time $\tau_x \coloneqq \inf\set{t \in \N_0 | S_t = x}$ indicating when the random walk reaches $x$ for the first time.
Then, for any $a, b > 0$,
$\Pr{\tau_{+b} < \tau_{-a}} = (r^a - 1)/(r^{a + b} - 1)$.
\end{lemma}
\begin{proof}
Consider martingale $M_0 \coloneqq 1$ and
\begin{math}
M_i
\coloneqq
r^{X_i} \cdot M_{i-1}
\end{math}.
Define the stopping time $\tau \coloneqq \min\set{\tau_{-a}, \tau_b}$.
By the optional stopping theorem, we have $\Exp[M_{\tau}] = \Exp[M_0] = 1$.
This yields
$
1
=
\Exp[M_0]
=
\Exp[M_{\tau}]
=
\Pr{\tau_{+b} < \tau_{-a}} \cdot r^b + \bigl(1 - \Pr{\tau_{+b} < \tau_{-a}}\bigr) \cdot r^{-a}
$.
Rearranging for $\Pr{\tau_{+b} < \tau_{-a}}$ gives the desired result.
\end{proof}

The next \lcnamecref{lem:maintainvalidity} is our main tool in proving \cref{thm:maintainmaxdiscrepancy}.
It states that if started in a safe configuration $\bm{x}$, our process maintains safety (with respect to the base height $h(t) \coloneqq \ceil{m_{\max}(t)/n} + O(1)$) for a polynomial time.
The basic idea is to consider any time $t'$ in which some previously safe level $\ell$ just became critical.
Using \cref{lem:murw_drift}, we then couple the random variable $m_{\ell}(\bullet)$ with a random walk that is always at least as large as $m_{\ell}(\bullet)$ and is biased by a factor of exactly $2$ (or $\sqrt{n}$, in the case of level $\ell = \ell^* + 1$) towards the safe left side of the critical interval $C_{\ell}$.
Applying \cref{lem:biasedrw_crossprob} to all possible such random walks together with a union bound then implies that, \whp, no level crosses the critical interval and becomes invalid for $\poly(n)$ many steps.

\begin{restatable}{lemma}{restatemaintainvalidity}
\label{lem:maintainvalidity}
Let $h(t) \coloneqq \ceil{m_{\max}(t)/n} + \gamma$ denote the base height at time $t$.
Consider an initial configuration $\bm{x(0)} = \bm{x}$ that is $h(0)$-safe.
Fix the time $T = n^4$ and assume $n$ to be large enough.
If $\beta(t) \leq \hat{\beta} < 1$ for all $t \in \set{0, 1, \dots, T - 1}$, then
\begin{equation}
\Pr{\forall t \in \set{0, 1, \dots, T - 1}\colon \text{$\bm{x(t)}$ is $h(t)$-valid}}
\geq
1 - n^{-1}
.
\end{equation}
\end{restatable}

The proof of the lemma may be found in \ifthenelse{\isfullversion}{\cref{app:layered}}{the full version}.
It remains to prove \cref{thm:maintainmaxdiscrepancy}, which follows by applying \cref{lem:maintainvalidity} and observing that any valid configuration has the desired maximum load.

\begin{proof}[Proof of \cref{thm:maintainmaxdiscrepancy}]
Define the base height $h(t) \coloneqq \ceil{m_{\max}(t)/n} + \gamma$ at time $t$ as in \cref{lem:maintainvalidity}.
Consider an initial configuration $\bm{x} = \bm{x}(0)$ that is perfectly balanced (i.e., a configuration with maximum load $\ceil{m_{\max}(t)/n}$).
Note that $\bm{x}$ is trivially $h(0)$-safe, since (choosing $\gamma \geq 1$) there are no balls with height $\geq h(0)$.

Then, by \cref{lem:maintainvalidity}, all configurations $\bm{x(t)}$ with $t \in \intco{0, n^4}$ are (\whp) $h(t)$-valid.
Thus, it is sufficient to show that the maximum load in any $h$-valid configuration is $h + \log\log n + O(1)$.
This follows immediately from the definition of the critical thresholds:
Indeed, if $\bm{x}$ is $h$-valid, so is level $\ell^* + 1$.
Thus, there are at most $\alpha_{\ell^* + 1} = O(1)$ many balls on or above level $\ell^* + 1$.
In particular, the maximum load in $\bm{x}$ is
\begin{math}
x_{\max}
\leq
h + \ell^* + 1 + \alpha_{\ell^* + 1}
=
h + \log\log n + O(1)
\end{math}.
This proves the first part of the \lcnamecref{thm:maintainmaxdiscrepancy}.

For the second part, assume that the insertion probability in the interval $[0, n^4]$ is $c$-good for a constant $c$.
That is, for any discrete subinterval $(t_1, t_2] \subseteq [0, n^4]$ of length $t_2 - t_1 \geq c \cdot n$ we have
$(\sum_{t = t_1 + 1}^{t_2}\beta(t))/(t_2 - t_1)
\geq
\frac{1}{2} \cdot (1 + \epsilon)
$
for a constant $\epsilon > 0$.
Consider the random variable $Y = \sum_{t = t_1 + 1}^{t_2} Y_t$ where $Y_t = 1$ with probability $\beta(t)$ and $Y_t = 0$ otherwise.
Note that $Y$ counts the number of balls added in the time interval $\intoc{t_1, t_2}$.
By $c$-goodness we have
\begin{math}
\Exp[Y]
\geq
(1 + \epsilon) \cdot \frac{t_2 - t_1}{2}
\end{math}.
A simple application of Chernoff yields
\begin{equation}%
\label{eqn:maintainmaxdiscrepancy:tailboundaddedballs}
\Pr{Y < (1 + \epsilon/2) \cdot \frac{t_2 - t_1}{2}}
\leq
e^{-\Theta(1) \cdot (t_2 - t_1)}
=
n^{-\omega(1)}
.
\end{equation}
Note that if $Y$ balls are added during the interval $\intoc{t_1, t_2}$, at most $t_2 - t_1 - Y$ balls are deleted.
Thus, using \cref{eqn:maintainmaxdiscrepancy:tailboundaddedballs}, the total change of the system load during the interval $\intoc{t_1, t_2}$ is, with very high probability,
\begin{equation}
m(t_2) - m(t_1)
\geq
Y - (t_2 - t_1 - Y)
=
2Y - (t_2 - t_1)
\geq
\epsilon \cdot \frac{t_2 - t_1}{2}
>
0
.
\end{equation}
Thus, we can take a union bound over the $\poly(n)$ many subintervals $\intoc{t_1, t_2} \subseteq \intcc{0, n^4}$ of length at least $c \cdot n$ to get that, w.h.p., the average load at the end of the subinterval is at least the average load at its beginning.
All other subintervals have length $< c \cdot n$, so the average load can decrease by at most $c$ during any of them.
Together, we get that, w.h.p., for any $t \in \intcc{0, n^4}$ we have
\begin{math}
m(t)
\geq
m_{\max}(t) - c
\end{math}.
Thus, by using the \lcnamecref{thm:maintainmaxdiscrepancy}'s first statement we can, \whp, bound the maximum load at any time $t \in \intcc{0, n^4}$ by
\begin{math}
x_{\max}(t)
\leq
\ceil{m_{\max}(t)/n} + \log\log n + O(1)
\leq
\ceil{m(t)/n} + \log\log n + O(1) + c
=
\ceil{m(t)/n} + \log\log n + O(1).
\end{math}
\end{proof}

\section{Quick Recovery \& Steady State Characterization.}%
\label{sec:coupling}

With the tools from \cref{sec:larstom,sec:layered} at hand, we are now ready to prove our main result (\cref{thm:main3}).
It generalizes the guarantees given by \cref{thm:maintainmaxdiscrepancy} (there only for the first $\poly(n)$ many steps) to an \emph{arbitrary} time $t \in \N_0$.
\begin{theorem}%
\label{thm:main3}
Consider an insertion probability sequence $\intoo[\big]{\beta(t)}_{t \in \N}$ bounded away from $0$ and $1$.
Then, \whp,
\begin{math}
x_{\max}(t)
\leq
 m_{\max}(t)/n + \log\log n + O(1)
\end{math}
for any $t \in \N_0$.
Moreover, if there is a constant $c \in \N$ such that time interval $\intcc{t-n^4, t}$ is $c$-good, then, \whp,
\begin{math}
x_{\max}(t)
\leq
m(t)/n + \log\log n + O(1)
\end{math}.
\end{theorem}
We conjecture that the requirement for a \emph{constant} $\epsilon > 0$ in \cref{def:good} is an artifact of our analysis technique and can be reduced to $o(1)$ or even eliminated.
The same holds for our assumption that the insertion probability must be bounded away from $1$ (note that for $\beta = 1$ the result follows from \cite{BCSV06}).
We note that the first of the two results in the theorem also holds when first $m\ge n$ balls are inserted and the process then alternates between insertions and deletions.

\paragraph{Analysis Overview.}
The proof idea of \cref{thm:main3} is as follows:
For any time $t \in \N_0$, \cref{thm:potential_bound} implies (\whp) a logarithmic absolute discrepancy at time $t' \coloneqq t - n^4$.
We then use a path coupling argument similar to \cite{DBLP:journals/mst/Czumaj00} to show that our process causes \emph{any} load situation with logarithmic absolute discrepancy to recover quickly to a \enquote{typical} state.
More exactly, any two (different) such load situations at time $t'$ can be coupled such that (\whp) they become identical at time $t$.
At this point however, we don't know how such a \enquote{typical} looks like (and whether it has small maximal load).
This is where \cref{thm:maintainmaxdiscrepancy} comes into play.
It tells us that if we start in a perfectly balanced load situation at time $t'$, our process maintains (\whp) a double-logarithmic maximal load for \emph{at least} $n^4$ steps.
A simple union bound over these two high probability bounds then implies that the \enquote{typical} load situation at time $t$ must also have a double-logarithmic maximal load.

\paragraph{Analysis.}
We now formalize the above idea.
To this end, we first introduce a measure for the similarity of two given load situations.
\begin{definition}%
\label{def:transformationdistance}
Consider two load vectors $\bm{x}, \bm{y}$ over $n \in \N$ bins with identical total loads $\norm{\bm{x}}_1 = \norm{\bm{y}}_1$.
The \emph{transformation distance} between $\bm{x}$ and $\bm{y}$ is
\begin{math}
\Delta(\bm{x}, \bm{y})
\coloneqq
\norm{\bm{x} - \bm{y}}_1 / 2
\end{math}.
\end{definition}
Note that the transformation distance obeys the identity
\begin{math}
\Delta(\bm{x}, \bm{y})
=
\sum_{i=1}^{n} \max\set{0, x_i - y_i}
\end{math}
and corresponds to the minimal number of ball movements to transform $\bm{x}$ into $\bm{y}$ (and vice versa).
Moreover, any load vector $\bm{x}$ can be transformed into the perfectly balanced load vector by moving the at most $n \cdot \disc(\bm{x}) \leq n \cdot \adisc(\bm{x})$ balls above average.
As an immediate consequence, we get the following observation.
\begin{obs}%
\label{obs:distfromdiscrepancy}
For two load vectors $\bm{x}, \bm{y}$ over $n \in \N$ bins with identical total loads $\norm{\bm{x}}_1 = \norm{\bm{y}}_1$, we have
\begin{math}
\Delta(\bm{x}, \bm{y})
\leq
2n \cdot \max\set{\adisc(\bm{x}), \adisc(\bm{y})}
\end{math}.
\end{obs}

The following \lcnamecref{lem:couplingneighbors} leverages a result from \cite{DBLP:journals/mst/Czumaj00} to show that one step of our process tends to reduce transformation distance between two given load vectors.
In~\cite{DBLP:journals/mst/Czumaj00}, the author considers basically the same setting as ours, but with guaranteed \emph{alternating} insertions/deletions.
\Cref{lem:couplingneighbors} basically shows that their proof transfers directly to our more general setting with insertion probability sequences.

\begin{lemma}%
\label{lem:couplingneighbors}
Consider two load vectors $\bm{x}, \bm{y}$ over $n \in \N$ bins with identical total loads $\norm{\bm{x}}_1 = \norm{\bm{y}}_1$.
Assume $\Delta(\bm{x}, \bm{y}) = 1$ and let $\bm{x'}$ and $\bm{y'}$ denote the load vectors after applying one step of our process with insertion probability $\beta \in \intcc{0, 1}$ to $\bm{x}$ and $\bm{y}$, respectively.
There is a coupling of $\bm{x'}$ and $\bm{y'}$ such that
\begin{enumerate}
\item\label{lem:couplingneighbors:a}
    \begin{math}
    \Exp[\Delta(\bm{x'}, \bm{y'})]
    \leq
    \Delta(\bm{x}, \bm{y})
    =
    1
    \end{math}
    and

\item\label{lem:couplingneighbors:b}
    \begin{math}
    \Pr{\Delta(\bm{x'}, \bm{y'}) \neq \Delta(\bm{x}, \bm{y})}
    \geq
    (1 - \beta) / n
    \end{math}.
\end{enumerate}
\end{lemma}

\begin{proof}
For the random decision whether to insert or delete a ball in $\bm{x}$ and $\bm{y}$, we use the identity coupling that either inserts a ball in both processes (with probability $\beta(t)$) or deletes a ball in both processes (with probability $1 - \beta(t)$).
It remains to couple the random bin choices in a deletion step and the two random bin choices of \Greedy{2} in an insertion step.

Assume the considered step is a deletion step.
Then we use the first coupling described in~\cite[Section~6]{DBLP:journals/mst/Czumaj00} (who consider \emph{alternating} insertions/deletions instead of our more flexible insertion probability sequences).
In this case, \cite[Claims~6.1 and~6.2]{DBLP:journals/mst/Czumaj00} together yield
\begin{math}
\Exp[\Delta(\bm{x'}, \bm{y'}) \mid \text{Deletion Step}]
\leq
\Delta(\bm{x}, \bm{y})
=
1
\end{math}
and
\begin{math}
\Pr{\Delta(\bm{x'}, \bm{y'}) \neq \Delta(\bm{x}, \bm{y}) | \text{Deletion Step}}
\geq
1 / n
\end{math}.
Note that this already yields the \lcnamecref{lem:couplingneighbors}'s second statement, since a deletion step occurs with probability $1 - \beta$.

Similarly, for an insertion step, we use the second coupling described in~\cite[Section~6]{DBLP:journals/mst/Czumaj00}, where the exact same insertion rule is used (again, in a setting with alternating insertions/deletions).
In this case, the proof of \cite[Lemma~5.1]{DBLP:journals/mst/Czumaj00} implies
\begin{math}
\Exp[\Delta(\bm{x'}, \bm{y'}) \mid \text{Insertion Step}]
\leq
\Delta(\bm{x}, \bm{y})
=
1
\end{math}.
Together, we get the desired result.
\end{proof}

The next \lcnamecref{lem:reflrwcrossingtime} characterizes the expected and high probability crossing time of a simple, $\alpha$-lazy reflecting random walk on $\set{0, 1, \dots, D}$ (i.e., a random walk that has the same probability to move left/right and that stays put with probability $\alpha$ or if it tries to move out of bounds). We use the lemma for our coupling result in \cref{lem:couplingarbitrary}.
\begin{lemma}%
\label{lem:reflrwcrossingtime}
Let $D \in \N_0$ and $\alpha \in \intco{0, 1}$.
Consider a simple, $\alpha$-lazy reflecting random walk $\intoo{W_t}_{t \in \N_0}$ on $\set{0, 1, \dots, D}$ started in $W_0 = D$.
Let $T \coloneqq \min\set{t \in \N | W_t = 0}$.
Then
\begin{math}
\Exp[T]
=
D \cdot (D + 1) / (1 - \alpha)
\end{math}.
Moreover, for any $n \in \N$ and $a > 0$ we have
\begin{math}
\Pr{T \geq a \cdot \log(n) \cdot 2\Exp[T]}
\leq
1/n^{-a}
\end{math}.
\end{lemma}
\begin{proof}
Let $E_i$ denote the expected time for the random walk to reach $0$ if started at position $i \in \set{0, 1, \dots, D}$.
Clearly, $E_0 = 0$ and $E_D = \Exp[T]$.
Note that we have the recurrence relation $E_i = 1 + E_i \cdot \alpha +  (E_{i-1} + E_{i+1}) \cdot (1 - \alpha)/2$ for $i \in \set {1, 2, \dots, D-1}$ and $E_D = 1 + E_{D-1} \cdot (1 - \alpha)/2 + E_D \cdot (1 + \alpha)/2$.
From this one can first deduce $E_i = 2(D-i + 1)/(1 - \alpha) + E_{i-1}$ for $i \in \set{1, 2, \dots, D}$ and, then, $\Exp[T] = E_D = D \cdot (D + 1) / (1 - \alpha)$.

For the high probability bound, we observe that Markov's inequality implies $\Pr{T \geq 2\Exp[T]} \leq 1/2$.
Thus, after at most $a \cdot \log(n)$ repetitions of such phases of length $2\Exp[T]$, the random walk has reached $0$ with probability at least ${(1/2)}^{a \cdot \log(n)} = n^{-a}$.
\end{proof}

The following lemma shows  that it is possible to couple two instances of our process with a logarithmic absolute discrepancy; they will be in the same distribution after polynomial time.

\begin{lemma}%
\label{lem:couplingarbitrary}
Consider two instances of our random processes $\intoo[\big]{\bm{x(t)}}_{t \in \N_0}$ and $\intoo[\big]{\bm{y(t)}}_{t \in \N_0}$ with
\begin{math}
\norm{\bm{x(0)}}_1
=
\norm{\bm{y(0)}}_1
\end{math}
and
\begin{math}
\adisc(\bm{x(0)}), \adisc(\bm{y(0)})
\leq
\frac{3}{\alpha} \cdot \log n
\end{math}.
Assume both processes use the same insertion probability sequence $\intoo[\big]{\beta(t)}_{t \in \N_0}$ bounded away from $0$ and $1$.
Then there is a coupling of $\intoo[\big]{\bm{x(t)}}_{t \in \N_0}$ and $\intoo[\big]{\bm{y(t)}}_{t \in \N_0}$ for whose coupling time
\begin{math}
\tau
\coloneqq
\min\Set{ t \in \N_0 | \Delta\intoo[\big]{\bm{x(t)}, \bm{y(t)}} = 0  }
\end{math}
we have, \whp, $\tau = O\intoo[\big]{n^3 \cdot {(\log n)}^3}$.
\end{lemma}
\begin{proof}
Consider any time $t \in \N_0$ and define
\begin{math}
\Delta(t)
\coloneqq
\Delta\intoo[\big]{\bm{x(t)}, \bm{y(t)}}
\end{math}.
There is a sequence $\bm{x(t)} = \bm{x^{(0)}(t)}, \bm{x^{(1)}(t)}, \dots, \bm{x^{(\Delta(t))}(t)} = \bm{y(t)}$ of load vectors such that for all indices $i \in \set{1, 2, \dots, \Delta(t)}$ we have
\begin{math}
\Delta^{(i)}(t)
\coloneqq
\Delta\intoo*{\bm{x^{(i-1)}(t)}, \bm{x^{(i)}(t)}}
=
1
\end{math}.

By \cref{lem:couplingneighbors}\ref{lem:couplingneighbors:a}, there is a coupling of each pair $\intoo*{\bm{x^{(i-1)}(t+1)}, \bm{x^{(i)}(t+1)}}$ such that $\Exp\intcc[\big]{\Delta^{(i)}(t+1)} \leq 1$.
This implies a coupling of $\bm{x(t+1)}$ and $\bm{y(t+1)}$ with
\begin{equation}%
\label{eqn:delta:expdecrease}
\Exp\intcc[\big]{\Delta(t+1)}
\leq
\Exp\intcc*{\sum_{i=1}^{\Delta(t)} \Delta^{(i)}(t+1)}
=
\sum_{i=1}^{\Delta(t)} \Exp\intcc*{\Delta^{(i)}(t+1)}
\leq
\Delta(t)
.
\end{equation}
Similarly, \cref{lem:couplingneighbors}\ref{lem:couplingneighbors:b} implies
\begin{equation}%
\label{eqn:delta:likelychange}
\Pr{\Delta(t+1) \neq \Delta(t) \;|\; \Delta(t) > 0}
\geq
\frac{1 - \beta(t)}{n}
.
\end{equation}

Consider the random process
\begin{math}
\intoo[\big]{\Delta(t)}_{t \in \N_0}
\end{math}.
Let
\begin{math}
D
\coloneqq
n \cdot \frac{10}{\alpha} \cdot \log n
=
O(n \cdot \log n)
\end{math}
and define the stopping time
\begin{equation}
T_1
\coloneqq
\min\set{ t \in \N_0 | \Delta(t) > D }
.
\end{equation}
Let $\beta \coloneqq \sup\set{\beta(t) | t \in \N_0} < 1$ and $\alpha \coloneqq 1 - (1 - \beta) / n$.
Because of \cref{eqn:delta:expdecrease,eqn:delta:likelychange}, during the time steps $t \in \set{0, 1, \dots, T_1-1}$ the distance $\Delta(t)$ is majorized by the position $W(t)$ of a simple, $\alpha$-lazy reflecting random walk
\begin{math}
\intoo[\big]{W(t)}_{t \in \N_0}
\end{math}
on $\set{0, 1, \dots, D}$ starting in $W(0) = D$.
Define the stopping time
\begin{equation}
T_2
\coloneqq
\min\set{ t \in \N_0 | W(t) = 0 }
.
\end{equation}
Note that if $T_1 > T_2$, the majorization implies
\begin{math}
\Delta(T_2)
\leq
W(T_2)
=
0
\end{math}
and, thus, $\tau \leq T_2$.

By applying \cref{thm:potential_bound} with $a = 5$ to the first $n^4$ time steps and using a union bound, \whp~both $\intoo[\big]{\bm{x(t)}}_{t \in \N_0}$ and $\intoo[\big]{\bm{y(t)}}_{t \in \N_0}$ maintain an absolute discrepancy of at most $\adisc(0) + \frac{7}{\alpha} \cdot \log n = \frac{10}{\alpha} \cdot \log n$ during the first $n^4$ steps.
By \cref{obs:distfromdiscrepancy}, this implies a transformation distance of at most $n \cdot \frac{10}{\alpha} \cdot \log n = D$ and, thus, $T_1 > n^4$.
By \cref{lem:reflrwcrossingtime}, \whp, $T_2 = O\intoo[\big]{n^3 \cdot {(\log n)}^3}$.
Thus, by a union bound we have, \whp, $T_1 \geq T_2$ and $T_2 = O\intoo[\big]{n^3 \cdot {(\log n)}^3}$.
As argued above, together these imply
\begin{math}
\tau
\leq
T_2
=
O\intoo[\big]{n^3 \cdot {(\log n)}^3}
\end{math}.
\end{proof}

With the coupling at hand we are ready to prove our main result.

\begin{proof}[Proof of \cref{thm:main3}]
Note that for $t \leq n^4$, the statement follows immediately from $\cref{thm:maintainmaxdiscrepancy}$.
For larger $t$, \cref{thm:potential_bound} gives us, \whp, $x_{\max}(t-n^4) - x_{\min}(t-n^4) \leq 2c \cdot \log n$.
By this and by shifting each time step forward by $n^4$ steps, it is sufficient to prove the \lcnamecref{thm:main3}'s statement for $t = n^4$ starting from an initial load vector $\bm{x(0)}$ with $x_{\max}(0) - x_{\min}(0) \leq 2c \cdot \log n$.

Define an initial load vector $\bm{y(0)}$ with total load $\norm{\bm{y(0)}}_1 = m(0)$ that is perfectly balanced.
By \cref{lem:couplingarbitrary} there is a coupling of $\intoo[\big]{\bm{x(t)}}_{t \in \N_0}$ and $\intoo[\big]{\bm{y(t)}}_{t \in \N_0}$ such that, \whp, $x(t) = y(t)$ for all $t = \tilde{\Omega}(n^3)$ (so, in particular, for $t = n^4$).
Note that \cref{thm:maintainmaxdiscrepancy} applies to $\intoo[\big]{\bm{y(t)}}_{t \in \N_0}$.
This gives us the corresponding high probability guarantees from \cref{thm:maintainmaxdiscrepancy} for $\bm{y(n^4)}$ which, by leveraging the coupling, hold \whp~for $\bm{x(n^4)}$, finishing the proof.
\end{proof}

\section*{Acknowledgments.}
Petra Berenbrink's research was funded by the DFG project \emph{Distributed and Collaborative Systems of Agents} (project number 411362735) as part of the Research Unit (Forschungsgruppe) \emph{Algorithms, Dynamics, and Information Flow in Networks} (FOR 2975).

\bibliographystyle{alpha}
\bibliography{literature}

\ifthenelse{\isfullversion}{
    \pagebreak
    \appendix
    \section*{Appendix.}

    \section{Deletion of a Random Ball.}%
\label{sec:randomball}

Consider an insertion probability sequence $\intoo[\big]{\beta(t)}_{t \in \N}$.
Let $X = (\bm{x}(t))_{t\in N_0}$ be the process where balls are deleted from randomly chosen non-empty bins during a deletion step.
Remember that configuration $\bm{x}(t)$ is the load vector after step $t$ with $x_i(t)$ denoting the $i$-th fullest bin.
Similarly, let $Y = (\bm{y}(t))_{t\in N_0}$ be the process in which every ball is deleted with the same probability during a deletion step.

Our goal is to show that process $Y$ is not worse (w.r.t.~discrepancy/overload) than process $X$.
We formalize this using the standard notion of \emph{majorization}:
For a configuration $\bm{x}$ let $S_k(\bm{x}) = \sum_{i=1}^k x_i$ be the total load of the $k$ fullest bins.
We say a configuration $\bm{x}$ \emph{majorizes} configuration $\bm{y}$ if $S_k(\bm{y}) \leq S_k(\bm{x})$ for all $k \in \set{1, 2, \dots, n}$.
Note that this immediately implies that $\bm{y}$ has a smaller maximal load than $\bm{x}$.
Thus, we can formalize the desired result via the following \lcnamecref{lem:majorization}:
\begin{theorem}%
\label{lem:majorization}
Starting from an empty system, there exists a coupling between processes $X$ and $Y$ such that $\bm{x}(t)$ majorizes $\bm{y}(t)$ for all $t \in \N_0$.
\end{theorem}

\paragraph{Coupling.}
Let us define a suitable coupling between processes $X$ and $Y$.
To this end, we use the identity coupling for the decision whether a ball is allocated (with probability $\beta(t)$) or deleted (with probability $1 - \beta(t)$).
As a result, the total load in both processes remains the same at any time; we denote it by $m(t)$.
It remains to couple the random choices for \emph{allocation} and for \emph{deletion} steps from configuration $\bm{x}(t)$ to configuration $\bm{x}(t+1)$:
\begin{description}
\item[Allocation Step:]
    Here we use the identity coupling.
    That is, both processes use the same two random bin choices for \Greedy{2}.

\item[Deletion Step:]
    Here we number the $m(t)$ balls in both configurations $\bm{x}(t)$ and $\bm{y}(t)$ from left to right (and arbitrarily within each bin).
    We then draw a uniformly random number $z(t)$ from $[0,1)$ and use it to determine the ball deletions as follows:
    \begin{enumerate}
    \item $Y$ deletes the ball $i \in \set{1, \dots, m(t)}$ if and only if
        \begin{equation}
        \frac{i-1}{m(t)}
        \leq
        z(t)
        <
        \frac{i}{m(t)}
        .
        \end{equation}

    \item $X$ deletes the ball $j \in \set{1, \dots, m(t)}$, where $j = \sum_{a=1}^{\ell-1} x_a(t) + k$ (i.e., the $k$-th ball in bin $\ell$), if and only if
        \begin{equation}
        \frac{\ell-1}{\hat{n}(t)} + \frac{k-1}{\hat{n}(t) \cdot x_{\ell}(t)}
        \leq
        z(t)
        <
        \frac{\ell-1}{\hat{n}(t)}+\frac{k}{\hat{n}(t) \cdot x_{\ell}(t)}
        ,
        \end{equation}
        where $\hat{n}(t)$ denotes the number of non-empty bins in configuration $\bm{x}(t)$.
    \end{enumerate}
\end{description}
Note that our coupling of the deletions maintains the marginal distributions of $X$ and $Y$.
Indeed, process $Y$ deletes ball $i$ with probability exactly $1/m$.
Similarly, process $X$ deletes ball $j$ with probability exactly $1/(\hat{n}(t) \cdot x_{\ell}(t))$ (a uniformly random ball chosen from a uniformly random non-empty bin).

\paragraph{Analysis.}
In the prove of \cref{lem:majorization}, we make use of the following simple observation:
\begin{obs}\label{obs:rs}
Consider the above coupling for processes $X$ and $Y$ during a deletion step $t + 1$.
Assume that $\bm{y}(t)$ is majorized by $\bm{x}(t)$.
Let $r$ and $s$ denote the bins from which $X$ and $Y$ delete during step $t+1$, respectively.
Similarly, let $j$ and $i$ denote the number of the balls deleted by $X$ and $Y$ during step $t+1$, respectively.
Then
\begin{equation}
S_{s-1}(\bm{y}(t))
<
i
\leq
j
\leq
S_r(\bm{x}(t))
\end{equation}
\end{obs}
The first inequality follows by our ball numbering scheme (from left to right), since $i \leq S_{s-1}(\bm{y}(t))$ would imply that ball $i$ lies in a bin $< s$.
Similarly, the last inequality follows since $j > S_r(\bm{x}(t))$ would imply that ball $j$ lies in a bin $> r$.
The inequality $i \leq j$ is a simple consequence of our coupling, since the probabilities for ball $i$ being deleted in $Y$ are constant (w.r.t.~$i$) while the probabilities for ball $j$ being deleted in $X$ are monotonously non-decreasing (w.r.t.~$j$).

With this, we are ready to prove the main result of this section.
\begin{proof}[Proof of \cref{lem:majorization}]
We use the coupling defined above and prove the statement via induction over $t$, with the trivial base case $S_k(\bm{x}(0)) = 0 = S_k(\bm{y}(0))$ for all $k$.
For the inductive step we show that the coupling maintains majorization when going from $\bm{x}(t)$ to $\bm{x}(t+1)$ during time step $t+1$.
If a ball is inserted during time step $t+1$, this follows immediately from~\cite[Lemma~3.4]{ABKU99}.
Thus, in the rest of this proof we consider a deletion during time step $t+1$.
Moreover, as in \cref{obs:rs}, we use
\begin{enumerate*}
\item $r$ and $s$ to denote the bins from which $X$ and $Y$ delete a ball in step $t$, respectively; and
\item $j$ and $i$ to denote the balls deleted by $X$ and $Y$, respectively.
\end{enumerate*}

A \emph{plateau} of a configuration $\bm{x}$ is a maximal set of consecutive bins having the same load.
Note that when going from $\bm{x}(t)$ to $\bm{x}(t+1)$ by deleting a ball from bin $r$, we have to re-sort the bins after the deletion.
This will swap bin $r$ with the rightmost bin $r'$ of the plateau in which $r$ lies.
In other words, we can think of $X$ as deleting a ball from some bin $r' \geq r$ instead of bin $r$.
Formally we get $\bm{x}(t+1) = \bm{x}(t) - \bm{e_{r'}}$, where $\bm{e_i}$ denotes the $i$-th unit vector in $\R^n$.
Similarly, $Y$ deletes a ball from the rightmost bin $s' \geq s$ in the plateau containing bin $s$, yielding $\bm{y}(t+1) = \bm{y}(t) - \bm{e_{s'}}$

So assume $\bm{x}(t)$ majorizes $\bm{y}(t)$.
That is, $S_k(\bm{y}(t)) \leq S_k(\bm{x}(t))$ for all $k \in \set{1, 2, \dots, n}$.
We show by case distinction that $S_k(\bm{y}(t+1)) \leq S_k(\bm{x}(t+1))$ for all $k \in \set{1, 2, \dots, n}$:
\begin{description}
\item[Case~1: $k < \min\set{r', s'}$]
    Then the first $k$ bins remain untouched by the deletion in both processes, yielding
    \begin{equation}
    S_k(\bm{y}(t+1))
    =
    S_k(\bm{y}(t))
    \leq
    S_k(\bm{x}(t))
    =
    S_k(\bm{x}(t+1))
    .
    \end{equation}

\item[Case~2: $k \geq \max\set{r', s'}$]
    Then the first $k$ bins both loose exactly one ball in both processes, yielding
    \begin{equation}
    S_k(\bm{y}(t+1))
    =
    S_k(\bm{y}(t)) - 1
    \leq
    S_k(\bm{x}(t)) - 1
    =
    S_k(\bm{x}(t+1))
    .
    \end{equation}

\item[Case~3: $s' \leq k < r'$]
    Then $X$ and $Y$ delete exactly zero and one ball in the first $k$ bins, respectively, yielding
    \begin{equation}
    S_k(\bm{y}(t+1))
    =
    S_k(\bm{y}(t)) - 1
    \leq
    S_k(\bm{x}(t)) - 1
    =
    S_k(\bm{x}(t+1)) - 1
    <
    S_k(\bm{x}(t+1))
    .
    \end{equation}

\item[Case~4: $r' \leq k < s \leq s'$]
    Then we use \cref{obs:rs} to get
    \begin{equation}
    \begin{aligned}
    S_k(\bm{y}(t+1))
    =
    S_k(\bm{y}(t))
    &\leq
    S_{s-1}(\bm{y}(t))
      \\&\leq
    S_r(\bm{x}(t)) - 1
    \leq
    S_{r'}(\bm{x}(t)) - 1
    =
    S_{r'}(\bm{x}(t+1))
    \leq
    S_k(\bm{x}(t+1))
    .
    \end{aligned}
    \end{equation}

\item[Case~5: $\max\set{r', s} \leq k < s'$]
    For the sake of a contradiction, let $k$ be the minimal value in this range that violates the majorization after the deletion (and normalization).
    That is, we have
    \begin{math}
    S_k(\bm{y}(t+1))
    >
    S_k(\bm{x}(t+1))
    \end{math}
    and
    \begin{math}
    S_{k-1}(\bm{y}(t+1))
    \leq
    S_{k-1}(\bm{x}(t+1))
    \end{math}.\footnote{
        For $k = s$ this holds by the previous case.
    }
    This implies
    \begin{math}
    y_k(t+1)
    >
    x_k(t+1)
    \end{math}.
    Now note that $k$ lies in the plateau of $\bm{y}(t)$ between $s$ and $s'$, such that for all $q \in \set{k+1, k+2, s'-1}$ we have
    \begin{equation}
    y_q(t+1)
    =
    y_k(t+1)
    >
    x_k(t+1)
    \geq
    x_q(t+1)
    \end{equation}
    (the last inequality using that $\bm{x}(t+1)$ is sorted non-increasingly) and, similarly,
    \begin{equation}
    y_{s'}(t+1)
    =
    y_k(t+1) - 1
    \geq
    x_k(t+1)
    \geq
    x_q(t+1)
    .
    \end{equation}
    This implies
    \begin{math}
    S_{s'}(\bm{y}(t))
    =
    S_{s'}(\bm{y}(t+1)) + 1
    >
    S_{s'}(\bm{x}(t+1)) + 1
    =
    S_{s'}(\bm{x}(t))
    \end{math},
    contradicting the induction hypothesis.
\end{description}
Together, this shows that $\bm{x}(t+1)$ majorizes $\bm{y}(t+1)$.
\end{proof}

    \section{Omitted Material from Proof of Theorem \ref{thm:potential_bound}.}
\label{app:potential_bound}

In this section we prove \cref{thm:potential_bound}, which we restate for convenience:

\potentialbound*

We use the potential functions from \cite{DBLP:journals/rsa/PeresTW15} and show that, at an arbitrary point of time, the expected potential is linear in $n$.
Since the authors of \cite{DBLP:journals/rsa/PeresTW15}
considered a process without ball deletions, we extend the potential function analysis to capture ball deletions.
Recall that the process works as follows: In every time step a ball is inserted with probability $\beta(t)$ using two choices.
Otherwise, with probability $1-\beta(t)$, a random non-empty bin is picked and a ball from that bin is deleted.

The three potential functions defined in \cite{DBLP:journals/rsa/PeresTW15} are:

$$\Phi_{\alpha}(x(t)) = \sum_{i=1}^n e^{\alpha \cdot (x_i(t) - m(t)/n)};\hspace{1cm}
\Psi_{\alpha}(x(t)) = \sum_{i=1}^n e^{-\alpha \cdot (x_i(t) - m(t)/n)};$$
$$
\Gamma_{\alpha}(x(t)) = \Phi_{\alpha}(x(t)) + \Psi_{\alpha}(x(t))$$

\noindent where $x(t) = (x_1(t), ..., x_n(t))$ is the normalized load vector.
For convenience, we use $y_i(t) := x_i(t) - m(t)/n$ for the difference between the load in the $i$th bin and the average.

The remainder of the section is dedicated to prove Lemma \ref{thm:potential_bound1} which directly implies \cref{thm:potential_bound}.
For this proof, we require insertion and deletion probabilities.
Given that a ball is deleted in round $t$, we denote by $q_i(t)$ the probability that it is removed from the $i$th bin.
And given that a ball is inserted, we denote by $p_i(t)$ the probability that it is inserted into the $i$th bin.
If no bin is empty, then $q_i(t) = 1/n$ because the protocol deletes a ball from a non-empty bin chosen uniformly at random.
If there are $s$ empty bins, then the empty bins, which come last in the normalized array, have probability $q_i(t) = 0$, and the remaining bins have probability $q_i(t) = 1/(n-s)$.
$p_i(t)$, on the other hand, is independent of $t$.
$p_i(t) = p_i = \left(\frac{i}{n}\right)^2 - \left(\frac{i-1}{n}\right)^2$ because for the $i$th bin to be chosen in the normalized array, the $i$th bin must be among the choices next to a bin with the same or a smaller index.
(W.l.o.g.\ we assume that if bins of equal load are chosen, the bin with the larger index receives the ball.)
From this we can derive the following simple observation:

\begin{obs} \label{obs:epsilon}
$\sum_{i \leq n/4} p_i \leq 1/4 - \epsilon$ and $\sum_{i \geq 3n/4} p_i \geq 1/4 + \epsilon$ for any $\epsilon \leq \frac{3}{16}$.
\end{obs}

\begin{proof}
The probability for any of the first $n/4$ bins receiving the new ball is
$$\sum_{i \leq n/4} p_i = \sum_{i \leq n/4} \left(\frac{i}{n}\right)^2 - \left(\frac{i-1}{n}\right)^2 = \left(\frac{n/4}{n}\right)^2 = \frac{1}{16}.$$

The probability for any of the last $n/4$ bins receiving the new ball is
$$\sum_{i > 3n/4} p_i = \sum_{i > 3n/4} \left(\frac{i}{n}\right)^2 - \left(\frac{i-1}{n}\right)^2 = \left(\frac{n}{n}\right)^2 - \left(\frac{3n/4}{n}\right)^2 = 1 - \left(\frac{3}{4}\right)^2 = \frac{7}{16}$$
\end{proof}

Let $\hat{y}(t) = (\hat{y}_1(t), \ldots, \hat{y}_n(t))$ denote the load vector above average before normalization.
By combining the probabilities $\beta$, $p_i$ and $q_i$, we can express $\hat{y}_i(t+1)$, i.e.\ the $i$th bin's load above average before normalization at $t+1$, as a function of its previous load $y_i(t)$:

\begin{equation*}
\hat{y}_i(t+1) =
\begin{cases}
  y_i + 1 - 1/n, & \text{with probability $\beta(t) \cdot p_i$ (ball is allocated to $i$th bin)},\\
  y_i - 1/n, & \text{with probability $\beta(t) \cdot (1 - p_i)$ (ball is allocated to other bin)},\\
  y_i - 1 + 1/n, & \text{with probability $(1 - \beta(t)) \cdot q_i$ (ball is deleted from $i$th bin)},\\
  y_i + 1/n, & \text{with probability $(1 - \beta(t)) \cdot (1 - q_i)$ (ball is deleted from other bin)}
\end{cases}
\end{equation*}

Let $\Phi_{\alpha}^i(x(t)) = e^{\alpha \cdot y_i(t)}$ and $\Psi_{\alpha}^i(x(t)) = e^{-\alpha \cdot y_i(t)}$ be the local potentials in bin $i$.
And let $\Delta \hat{\Phi}_\alpha^i := \Phi_\alpha^i(\hat{x}(t+1)) - \Phi_\alpha^i(x(t))$ and $\Delta \hat{\Psi}_\alpha^i := \Psi_\alpha^i(\hat{x}(t+1)) - \Psi_\alpha^i(x(t))$ be the local potential changes in the $i$th bin (before normalization).
Then:

\begin{equation*}
\Delta \hat{\Phi}_\alpha^i =
\begin{cases}
  e^{\alpha (y_i + 1 - 1/n)} - e^{\alpha y_i} = e^{\alpha y_i} \cdot (e^{\alpha (1 - 1/n)} - 1), & \text{with probability $\beta(t) \cdot p_i$},\\
  e^{\alpha (y_i - 1/n)} - e^{\alpha y_i} = e^{\alpha y_i} \cdot (e^{-\alpha/n} - 1), & \text{with probability $\beta(t) \cdot (1 - p_i)$},\\
  e^{\alpha (y_i - 1 + 1/n)} - e^{\alpha y_i} = e^{\alpha y_i} \cdot (e^{\alpha (-1 + 1/n)} - 1), & \text{with probability $(1 - \beta(t)) \cdot q_i$},\\
  e^{\alpha (y_i + 1/n)} - e^{\alpha y_i} = e^{\alpha y_i} \cdot (e^{\alpha/n} - 1), & \text{with probability $(1 - \beta(t)) \cdot (1 - q_i)$}
\end{cases}
\end{equation*}

\begin{equation*}
\Delta \hat{\Psi}_\alpha^i =
\begin{cases}
  e^{-\alpha (y_i + 1 - 1/n)} - e^{\alpha y_i} = e^{-\alpha y_i} \cdot (e^{-\alpha (1 - 1/n)} - 1), & \text{with probability $\beta(t) \cdot p_i$},\\
  e^{-\alpha (y_i - 1/n)} - e^{\alpha y_i} = e^{-\alpha y_i} \cdot (e^{\alpha/n} - 1), & \text{with probability $\beta(t) \cdot (1 - p_i)$},\\
  e^{-\alpha (y_i - 1 + 1/n)} - e^{\alpha y_i} = e^{-\alpha y_i} \cdot (e^{-\alpha (-1 + 1/n)} - 1), & \text{with probability $(1 - \beta(t)) \cdot q_i$},\\
  e^{-\alpha (y_i + 1/n)} - e^{\alpha y_i} = e^{-\alpha y_i} \cdot (e^{-\alpha/n} - 1), & \text{with probability $(1 - \beta(t)) \cdot (1 - q_i)$}
\end{cases}
\end{equation*}

We use these local changes to prove different bounds on the overall potential changes $\Delta \Phi_\alpha = \Phi_\alpha(x(t+1)) - \Phi_\alpha(x(t))$ (in \cref{lem:phi_basis}, \cref{lem:phi1}, \cref{phi_corollary} and \cref{lem:phi2}) and $\Delta \Psi_\alpha = \Psi_\alpha(x(t+1)) - \Psi_\alpha(x(t))$ (in \cref{lem:psi_basis}, \cref{lem:psi1}, \cref{psi_corollary} and \cref{lem:psi2}).

\begin{lemma}[similar to a part of Lemma A.5 in \cite{DBLP:conf/icalp/TalwarW14}] \label{lem:phi_basis}
If, for a constant $\beta$, $\beta(t) \geq \beta$ for all $t$ and if $\alpha \leq \frac{\epsilon \cdot \beta}{3} \leq \frac{1}{16}$,
then $$E[\Delta \Phi_\alpha \;|\; x(t)] \leq \sum_{i=1}^n e^{\alpha y_i} \left( \beta(t) \left(p_i (e^{\alpha} - 1) - \frac{\alpha}{n}\right) + (1-\beta(t)) \frac{1}{n} (e^{-\alpha} - 1 + \alpha) + \frac{2 \alpha^2}{n^2} \right).$$
\end{lemma}

\begin{proof}
We compute the expected potential difference:
\begin{eqnarray}
& & E[\Delta \Phi_\alpha \;|\; x(t)] \nonumber\\
& = & \sum_{i=1}^n \beta(t) p_i \cdot (e^{\alpha (1-1/n)}-1) \cdot e^{\alpha y_i} + \beta(t) (1-p_i) \cdot (e^{\alpha (-1/n)}-1) \cdot e^{\alpha y_i} \nonumber\\
&&+ (1-\beta(t)) q_i \cdot (e^{\alpha (-1+1/n)}-1) \cdot e^{\alpha y_i} + (1-\beta(t)) (1-q_i) \cdot (e^{\alpha (1/n)}-1) \cdot e^{\alpha y_i} \nonumber\\
& = & \sum_{i=1}^n \beta(t) e^{\alpha y_i} \left( p_i \cdot e^{-\alpha/n} \cdot (e^{\alpha} - 1) + e^{-\alpha/n} - 1 \right) \nonumber\\
&&+ (1-\beta(t)) e^{\alpha y_i} \left( q_i e^{\alpha/n} (e^{-\alpha} - 1) + e^{\alpha/n} - 1 \right) \nonumber\\
& \leq & \sum_{i=1}^n \beta(t) e^{\alpha y_i} \left( p_i (e^{\alpha} - 1) + e^{-\alpha/n} - 1 \right) + (1-\beta(t)) e^{\alpha y_i} \left( q_i (e^{-\alpha} - 1) + e^{\alpha/n} - 1 \right) \nonumber\\
& = & \sum_{i=1}^n e^{\alpha y_i} \left( \beta(t) p_i (e^{\alpha} - 1) + (1-\beta(t)) q_i (e^{-\alpha} - 1) + e^{\alpha/n} - 1 + \beta(t) (e^{-\alpha/n} - e^{\alpha/n}) \right) \nonumber\\
& \leq & \sum_{i=1}^n e^{\alpha y_i} \left( \beta(t) p_i (e^{\alpha} - 1) + (1-\beta(t)) q_i (e^{-\alpha} - 1) + \frac{\alpha}{n} + \frac{2 \alpha^2}{n^2} - \frac{2 \beta(t) \alpha}{n} \right) \nonumber\\
& \leq & \sum_{i=1}^n e^{\alpha y_i} \left( \beta(t) p_i (e^{\alpha} - 1) + (1-\beta(t)) \frac{1}{n} (e^{-\alpha} - 1) + \frac{\alpha}{n} + \frac{2 \alpha^2}{n^2} - \frac{2 \beta(t) \alpha}{n} \right) \label{eq:phi_altered}\\
& = & \sum_{i=1}^n e^{\alpha y_i} \left( \beta(t) \left(p_i (e^{\alpha} - 1) - \frac{\alpha}{n}\right) + (1-\beta(t)) \frac{1}{n} (e^{-\alpha} - 1 + \alpha) + \frac{2 \alpha^2}{n^2} \right) \nonumber 
\end{eqnarray}

Inequality (\ref{eq:phi_altered}) holds because, if no bin is empty, all $q_i = 1/n$, and if bins are empty, then the $q_i$ are decreasing while the term $(1-\beta(t)) \cdot (e^{-\alpha} - 1) \cdot e^{\alpha y_i}$ is increasing as $e^{-\alpha} - 1 < 0$.
\end{proof}

\begin{lemma}[similar to a part of Lemma A.5 in \cite{DBLP:conf/icalp/TalwarW14}] \label{lem:phi1}
If, for a constant $\beta$, $\beta(t) \geq \beta$ for all $t$, if $\alpha \leq \frac{\epsilon \cdot \beta}{3} \leq \frac{1}{16}$ and if $y_{3n/4}(t) \leq 0$, then $E[\Delta \Phi_\alpha \;|\; x(t)] \leq - \frac{\beta \cdot \alpha \cdot \epsilon}{n} \cdot \Phi_\alpha + \alpha$.
\end{lemma}

\begin{proof}
We start from Lemma \ref{lem:phi_basis} and continue the calculation:
\begin{eqnarray}
& & E[\Delta \Phi_\alpha \;|\; x(t)] \nonumber\\
& \leq & \sum_{i=1}^n e^{\alpha y_i} \left( \beta(t) \left(p_i (e^{\alpha} - 1) - \frac{\alpha}{n}\right) + (1-\beta(t)) \frac{1}{n} (e^{-\alpha} - 1 + \alpha) + \frac{2 \alpha^2}{n^2} \right) \nonumber\\
& \leq & \sum_{i<3n/4} e^{\alpha y_i} \left( \beta(t) \left(p_i (e^{\alpha} - 1) - \frac{\alpha}{n} \right) + (1-\beta(t)) \frac{1}{n} (e^{-\alpha} - 1 + \alpha) + \frac{2 \alpha^2}{n^2} \right) \nonumber\\
&& + \sum_{i \geq 3n/4} e^{0} \left( \beta(t) \left(p_i (e^{\alpha} - 1) - \frac{\alpha}{n} \right) + (1-\beta(t)) \frac{1}{n} (e^{-\alpha} - 1 + \alpha) + \frac{2 \alpha^2}{n^2} \right) \nonumber\\
& \leq & \sum_{i<3n/4} \frac{4 \Phi_\alpha}{3 n} \left( \beta(t) \left(p_i (e^{\alpha} - 1) - \frac{\alpha}{n} \right) + (1-\beta(t)) \frac{1}{n} (e^{-\alpha} - 1 + \alpha) + \frac{2 \alpha^2}{n^2} \right) \label{eq:A5_0}\\
&& + \beta(t) \left(e^{\alpha} - 1 - \frac{\alpha}{4} \right) + (1-\beta(t)) \frac{1}{4} (e^{-\alpha} - 1 + \alpha) + \frac{\alpha^2}{2n} \nonumber\\
& \overset{(Obs.~\ref{obs:epsilon})}{\leq} & \frac{4 \Phi_\alpha}{3 n} \left( \beta(t) \left( \left(\frac{3}{4} - \epsilon\right) (e^{\alpha} - 1) - \frac{3 \alpha}{4} \right) + (1-\beta(t)) \frac{3}{4} (e^{-\alpha} - 1 + \alpha) + \frac{3 \alpha^2}{2 n} \right) + \alpha  \label{eq:A5_1}\\
& = & \frac{\Phi_\alpha}{n} \left( \beta(t) \left( \left(1 - \frac{4}{3} \epsilon\right) (e^{\alpha} - 1) - \alpha \right) + (1-\beta(t)) (e^{-\alpha} - 1 + \alpha) + \frac{2 \alpha^2}{n} \right) + \alpha \nonumber\\
& = & \frac{\Phi_\alpha}{n} \left( - \frac{4}{3} \beta(t) \epsilon (e^{\alpha} - 1) + \beta(t) \left( e^{\alpha} - 1 - \alpha \right) + (1-\beta(t)) (e^{-\alpha} - 1 + \alpha) + \frac{2 \alpha^2}{n} \right) + \alpha \nonumber\\
& \leq & \frac{\Phi_\alpha}{n} \left( - \frac{4}{3} \beta(t) \epsilon \alpha + \frac{2}{3} \alpha^2 \right) + \alpha \nonumber\\
& \leq & \frac{\Phi_\alpha}{n} \left( - \beta(t) \epsilon \alpha \right) + \alpha \nonumber\\
& \leq & \frac{\Phi_\alpha}{n} \left( - \beta \epsilon \alpha \right) + \alpha \nonumber
\end{eqnarray}

In Inequality~(\ref{eq:A5_0}), we use that $\sum_{i=1}^{n} p_i (e^\alpha-1) e^{\alpha y_i}$ is maximal when $e^{\alpha y_i} = \frac{4 \cdot \Phi_\alpha}{3 \cdot n}$ for all $i < 3n/4$.
This is because the $p_i$ are non-decreasing and the $y_i$ are non-increasing.
In Inequality~(\ref{eq:A5_1}), we use Observation~\ref{obs:epsilon} and the fact that $\sum_{i=1}^{n} p_i (e^\alpha-1) e^{\alpha y_i}$ is maximal if the $p_i$ are uniform.
The second to last inequality uses $\alpha \leq \beta \cdot \epsilon / 3 \leq \beta(t) \cdot \epsilon / 3$.
\end{proof}

\begin{corollary} \label{phi_corollary}
If, for a constant $\beta$, $\beta(t) \geq \beta$ for all $t$, if $\alpha \leq \frac{\epsilon \cdot \beta}{3} \leq \frac{1}{16}$ and if $y_{3n/4}(t) \leq 0$, then $E[\Delta \Phi_\alpha \;|\; x(t)] \leq \Phi_\alpha \cdot \frac{\alpha^2}{n}$.
\end{corollary}

\begin{proof}
We use Lemma \ref{lem:phi_basis} and continue the calculation:

\begin{eqnarray*}
& & E[\Delta \Phi_\alpha \;|\; x(t)] \nonumber\\
& \leq & \sum_{i=1}^n e^{\alpha y_i} \left( \beta(t) \left(p_i (e^{\alpha} - 1) - \frac{\alpha}{n}\right) + (1-\beta(t)) \frac{1}{n} (e^{-\alpha} - 1 + \alpha) + \frac{2 \alpha^2}{n^2} \right)\\
& \leq & \sum_{i=1}^n \frac{e^{\alpha y_i}}{n} \left( \beta(t) (e^{\alpha} - 1 - \alpha) + (1-\beta(t)) (e^{-\alpha} - 1 + \alpha) + \frac{2 \alpha^2}{n} \right) \nonumber\\
& \leq & \Phi_\alpha \cdot \frac{\alpha^2}{n} 
\end{eqnarray*}

The second inequality, for which all $p_i$ are set to $1/n$, holds because the $p_i$ are non-decreasing and the $e^{\alpha y_i}$ are non-increasing.
The last inequality makes use of the Taylor series.
\end{proof}

\begin{lemma}[similar to a part of Lemma A.6 in \cite{DBLP:conf/icalp/TalwarW14}] \label{lem:psi_basis}
If, for a constant $\beta$, $\beta(t) \geq \beta$ for all $t$ and if $\alpha \leq \frac{\epsilon \cdot \beta}{3} \leq \frac{1}{16}$, then
$$
E[\Delta \Psi_\alpha \;|\; x(t)] \leq \sum_{i=1}^n e^{-\alpha y_i} \left( \beta(t) (p_i (e^{-\alpha} - 1) + \alpha/n) + (1-\beta(t)) \frac{1}{n} (e^\alpha-1-\alpha) + \frac{2\alpha^2}{n^2} \right).
$$
\end{lemma}

\begin{proof}
We upper-bound the expected potential difference.

\begin{eqnarray}
& & E[\Delta \Psi_\alpha \;|\; x(t)] \nonumber\\
& = & \sum_{i=1}^n \beta(t) p_i \cdot (e^{-\alpha (1-1/n)}-1) \cdot e^{-\alpha y_i} + \beta(t) (1-p_i) \cdot (e^{-\alpha (-1/n)}-1) \cdot e^{-\alpha y_i} \nonumber\\
&&+ (1-\beta(t)) q_i \cdot (e^{-\alpha (-1+1/n)}-1) \cdot e^{-\alpha y_i} + (1-\beta(t)) (1-q_i) \cdot (e^{-\alpha (1/n)}-1) \cdot e^{-\alpha y_i} \nonumber\\
& = & \sum_{i=1}^n \beta(t) e^{-\alpha y_i} \left( p_i \cdot e^{\alpha/n} \cdot  (e^{-\alpha} - 1) + e^{\alpha/n}-1 \right) \nonumber\\
&&+ (1-\beta(t)) e^{-\alpha y_i} \left( q_i \cdot e^{-\alpha/n} \cdot (e^\alpha-1) + e^{-\alpha/n}-1 \right) \nonumber\\
& \leq & \sum_{i=1}^n \beta(t) e^{-\alpha y_i} \left( p_i (e^{-\alpha} - 1) + e^{\alpha/n}-1 \right) + (1-\beta(t)) e^{-\alpha y_i} \left( q_i (e^\alpha-1) + e^{-\alpha/n}-1 \right) \nonumber\\
& = & \sum_{i=1}^n e^{-\alpha y_i} \left( \beta(t) p_i (e^{-\alpha} - 1) + (1-\beta(t)) q_i (e^\alpha-1) + e^{-\alpha/n}-1 + \beta(t) (e^{\alpha/n} - e^{-\alpha/n}) \right) \nonumber\\
& \leq & \sum_{i=1}^n e^{-\alpha y_i} \left( \beta(t) p_i (e^{-\alpha} - 1) + (1-\beta(t)) q_i (e^\alpha-1) - \frac{\alpha}{n}  + \frac{2 \beta(t) \alpha}{n} + \frac{2\alpha^2}{n^2} \right) \nonumber\\
& \leq & \sum_{i=1}^n e^{-\alpha y_i} \left( \beta(t) p_i (e^{-\alpha} - 1) + (1-\beta(t)) \frac{1}{n} (e^\alpha-1) - \frac{\alpha}{n}  + \frac{2 \beta(t) \alpha}{n} + \frac{2\alpha^2}{n^2} \right) \label{eq:psi_altered}\\
& = & \sum_{i=1}^n e^{-\alpha y_i} \left( \beta(t) (p_i (e^{-\alpha} - 1) + \alpha/n) + (1-\beta(t)) \frac{1}{n} \left(e^\alpha - 1 - \alpha\right) + \frac{2\alpha^2}{n^2} \right) \nonumber
\end{eqnarray}

We use the Taylor series of the exponential function, and Inequality (\ref{eq:psi_altered}), replacing $q_i$ by $1/n$, holds because the $q_i$ are non-increasing and the $q^{-\alpha y_i}$ are non-decreasing.
\end{proof}

\begin{lemma}[similar to part of Lemma A.6 in \cite{DBLP:conf/icalp/TalwarW14}] \label{lem:psi1}
If, for a constant $\beta$, $\beta(t) \geq \beta$ for all $t$, if $\alpha \leq \frac{\epsilon \cdot \beta}{3} \leq \frac{1}{16}$ and if $y_{n/4} \geq 0$, then $E[\Delta \Psi_\alpha \;|\; x(t)] \leq - \frac{\beta \cdot \alpha \cdot \epsilon}{n} \cdot \Psi_\alpha + \alpha$.
\end{lemma}

\begin{proof}
We use Lemma \ref{lem:psi_basis} and continue the calculation:
\begin{eqnarray}
& & E[\Delta \Psi_\alpha \;|\; x(t)] \nonumber\\
& \leq & \sum_{i=1}^n e^{-\alpha y_i} \left( \beta(t) (p_i (e^{-\alpha} - 1) + \alpha/n) + (1-\beta(t)) \frac{1}{n} \left(e^\alpha - 1 - \alpha\right) + \frac{2\alpha^2}{n^2} \right) \nonumber\\
& \leq & \sum_{i=1}^n e^{-\alpha y_i} \left( \beta(t) \left(p_i (e^{-\alpha} - 1) + \frac{\alpha}{n}\right) + (1-\beta(t)) \frac{\alpha^2}{n} + \frac{2\alpha^2}{n^2} \right) \nonumber\\
& \leq & \sum_{i \geq n/4} e^{-\alpha y_i} \left( \beta(t) \left(p_i (e^{-\alpha} - 1) + \frac{\alpha}{n}\right) + (1-\beta(t)) \frac{\alpha^2}{n} + \frac{2\alpha^2}{n^2} \right) \nonumber\\
&& + \sum_{i < n/4} e^{0} \left( \beta(t) \left(p_i (e^{-\alpha} - 1) + \frac{\alpha}{n}\right) + (1-\beta(t)) \frac{\alpha^2}{n} + \frac{2\alpha^2}{n^2} \right) \nonumber\\
& \overset{(Obs.~\ref{obs:epsilon})}{\leq} & \sum_{i \geq n/4} \frac{e^{-\alpha y_i}}{n} \left( \beta(t) \left(1 + \frac{4}{3} \epsilon\right) (e^{-\alpha} - 1) + \beta(t) \alpha + (1-\beta(t)) \alpha^2 + \frac{2\alpha^2}{n} \right) \label{eq:A6_1}\\
&& + \frac{\beta(t)}{4} (e^{-\alpha} - 1 + \alpha) + \frac{1-\beta(t)}{4} \alpha^2 + \frac{\alpha^2}{2n}\nonumber\\
& \leq & \frac{\Psi_\alpha - n/4}{n} \left( \beta(t) \frac{4}{3} \epsilon (e^{-\alpha} - 1) + \frac{\alpha^2 + \alpha^3}{2}\right) + \frac{3}{4} \alpha^2 \label{eq:A6_2}\\
& \leq & \frac{\Psi_\alpha}{n} \left( \beta(t) \frac{4}{3} \epsilon (e^{-\alpha} - 1) + \alpha^2 \right) + \alpha \nonumber\\
& \leq & \frac{\Psi_\alpha}{n} \left( - \beta(t) \epsilon \alpha \right) + \alpha \nonumber\\
& \leq & \frac{\Psi_\alpha}{n} \left( - \beta \epsilon \alpha \right) + \alpha \nonumber
\end{eqnarray}

The tricky part is Eq. (\ref{eq:A6_1}) and (\ref{eq:A6_2}).
Note that $(e^{-\alpha}-1)$ is negative, and $p_i$ and $e^{-\alpha y_i}$ are weakly increasing.
So, to bound the expression, we can assume that the $p_i$ are uniform, and from $\sum_{i \geq n/4} p_i \geq 3/4 + \epsilon$ it follows that then $p_i \geq (3/4 + \epsilon) / (3/4 \cdot n) = (1 + 4/3 \cdot \epsilon) \cdot 1/n$ which implies Eq. (\ref{eq:A6_1}).
The resulting term $\frac{1}{n} \left( \beta(t) \left(1 + \frac{4}{3} \epsilon\right) (e^{-\alpha} - 1) + \beta(t) \alpha + (1-\beta(t)) (e^\alpha-1-\alpha) + \frac{2\alpha^2}{n} \right)$ is negative as well.
Since $\sum_{i<n/4} e^{-\alpha y_i} < n/4$ (because of $y_{n/4} \geq 0$), it follows that $\sum_{i \geq n/4} e^{-\alpha y_i} \geq \Psi_\alpha - n/4$ and thus Eq. (\ref{eq:A6_2}).

The second to last inequality uses $\epsilon \geq 3\alpha/\beta \geq 3\alpha/\beta(t)$.
\end{proof}

\begin{corollary} \label{psi_corollary}
If, for a constant $\beta$, $\beta(t) \geq \beta$ for all $t$, if $\alpha \leq \frac{\epsilon \cdot \beta}{3} \leq \frac{1}{16}$ and if $y_{n/4} \geq 0$, then $E[\Delta \Psi_\alpha \;|\; x(t)] \leq \Psi_\alpha \frac{\alpha^2}{n}$.
\end{corollary}

\begin{proof}
We use Lemma \ref{lem:psi_basis} and continue the calculation:
\begin{eqnarray*}
& & E[\Delta \Psi_\alpha \;|\; x(t)] \\
& \leq & \sum_{i=1}^n e^{-\alpha y_i} \left( \beta(t) (p_i (e^{-\alpha} - 1) + \alpha/n) + (1-\beta(t)) \frac{1}{n} (e^\alpha-1-\alpha) + \frac{2\alpha^2}{n^2} \right)\\
& \leq & \Psi_\alpha \frac{\alpha^2}{n}
\end{eqnarray*}
\end{proof}

\begin{lemma}[similar to Lemma A.7 in \cite{DBLP:conf/icalp/TalwarW14}] \label{lem:phi2}
Suppose $\beta(t) \geq \beta$, $y_{3n/4} > 0$ and $E[\Delta \Phi_\alpha \;|\; x(t)] \geq - \frac{\alpha \cdot \beta \cdot \epsilon}{4 \cdot n} \cdot \Phi_\alpha$.
Then either $\Phi_\alpha < \frac{\epsilon}{4} \cdot \Psi_\alpha$ or $\Gamma_\alpha \leq \frac{360 \cdot n}{\epsilon^8 \cdot \beta^4}$.
\end{lemma}

\begin{proof}
In the following we use that $p_i \leq \frac{1 - 4 \cdot \epsilon}{n}$ for $i \leq n/3$ and that $\sum p_i \cdot e^{\alpha \cdot y_i}$ is maximized when $\mathbf{p}$ is uniform.
We start from Lemma \ref{lem:phi_basis}.

\begin{eqnarray*}
& & E[\Delta \Phi_\alpha \;|\; x(t)] \\
& \leq & \sum_{i=1}^n e^{\alpha y_i} \left( \beta(t) (p_i (e^{\alpha} - 1) - \alpha/n) + (1-\beta(t)) \frac{1}{n} (e^{-\alpha} - 1 + \alpha) + \frac{2 \alpha^2}{n^2} \right)\\
& \leq & \sum_{i \leq n/3} e^{\alpha y_i} \left( \beta(t) \left(\frac{1-4 \epsilon}{n} (e^{\alpha} - 1) - \alpha/n\right) + (1-\beta(t)) \frac{1}{n} (e^{-\alpha} - 1 + \alpha) + \frac{2 \alpha^2}{n^2} \right)\\
&& + \sum_{i > n/3} e^{\alpha y_i} \left( \beta(t) \left(\frac{3}{2n} (e^{\alpha} - 1) - \alpha/n\right) + (1-\beta(t)) \frac{1}{n} (e^{-\alpha} - 1 + \alpha) + \frac{2 \alpha^2}{n^2} \right)\\
& \leq & \frac{\Phi_{\alpha, \leq n/3}}{n} \left( \beta(t) \left((1-4 \epsilon) (e^{\alpha} - 1) - \alpha\right) + (1-\beta(t)) (e^{-\alpha} - 1 + \alpha) + \frac{2 \alpha^2}{n} \right)\\
&& + \frac{\Phi_{\alpha, > n/3}}{n} \left( \beta(t) \left(\frac{3}{2} (e^{\alpha} - 1) - \alpha\right) + (1-\beta(t)) (e^{-\alpha} - 1 + \alpha) + \frac{2 \alpha^2}{n} \right)\\
& \leq & - \frac{3 \cdot \beta(t) \cdot \alpha \cdot \epsilon}{n} \cdot \Phi_\alpha + \frac{\alpha}{n} \cdot \Phi_{\alpha, > n/3}\\
& \leq & - \frac{3 \cdot \beta \cdot \alpha \cdot \epsilon}{n} \cdot \Phi_\alpha + \frac{\alpha}{n} \cdot \Phi_{\alpha, > n/3}
\end{eqnarray*}

Recall that we assume $E[\Delta \Phi_\alpha \;|\; x(t)] \geq - \frac{\alpha \cdot \beta \cdot \epsilon}{4 \cdot n} \cdot \Phi_\alpha$.

\begin{eqnarray*}
- \frac{\alpha \cdot \beta \cdot \epsilon}{4 \cdot n} \cdot \Phi_\alpha & \leq & - \frac{3 \cdot \alpha \cdot \beta \cdot \epsilon}{n} \cdot \Phi_\alpha + \frac{\alpha}{n} \cdot \Phi_{\alpha, > n/3}\\
\Rightarrow\;\;\;\; 2 \cdot \beta \cdot \epsilon \cdot \Phi_\alpha & \leq & \Phi_{\alpha, > n/3}
\end{eqnarray*}

If $\Phi_\alpha < \frac{\epsilon}{4} \cdot \Psi_\alpha$, we are done (see statement).
So, assume $\Phi_\alpha \geq \frac{\epsilon}{4} \cdot \Psi_\alpha$.
Since $\Phi_{\alpha, \geq n/3} \leq \frac{2 \cdot n}{3} \cdot e^{3 \cdot \alpha \cdot B / n}$ for $B = \sum_i \max\{0, y_i\} = \frac{||y||_1}{2}$ and since $y_{3n/4} > 0$ implies $\Psi_\alpha \geq \frac{n}{4} \cdot e^{4 \cdot \alpha \cdot B / n}$, we get:

\begin{eqnarray*}
\frac{n \cdot \epsilon}{16} \cdot e^{4 \cdot \alpha \cdot B / n} \leq \frac{\epsilon}{4} \cdot \Psi_\alpha \leq \Phi_\alpha & \leq & \Phi_{\alpha, > n/3} \cdot \frac{1}{2 \cdot \beta \cdot \epsilon} \leq \frac{n}{3 \cdot \beta \cdot \epsilon} \cdot e^{3 \cdot \alpha \cdot B / n}\\
\implies\;\;\;\; e^{\alpha \cdot B / n} \leq \frac{6}{\beta \cdot \epsilon^2}
\end{eqnarray*}

Finally, it follows that:
\begin{eqnarray*}
\Gamma_\alpha & = & \Phi_\alpha + \Psi_\alpha \leq \frac{5}{\epsilon} \cdot \Phi_\alpha \leq \frac{5}{\epsilon} \cdot \frac{n}{3 \cdot \beta \cdot \epsilon} \cdot e^{3 \cdot \alpha \cdot B / n} \leq \frac{5}{\epsilon} \cdot \frac{n}{3 \cdot \beta \cdot \epsilon} \cdot \left( \frac{6}{\beta \cdot \epsilon^2} \right)^3 = \frac{360 \cdot n}{\epsilon^8 \cdot \beta^4}
\end{eqnarray*}
\end{proof}

\begin{lemma}[similar to Lemma A.8 in \cite{DBLP:conf/icalp/TalwarW14}] \label{lem:psi2}
Suppose $\beta(t) \geq \beta$, $y_{n/4} < 0$ and $E[\Delta \Psi_\alpha \;|\; x(t)] \geq -\frac{\alpha \cdot \beta \cdot \epsilon}{4 \cdot n} \cdot \Psi_\alpha$.
Then either $\Psi_\alpha < \frac{\epsilon}{4} \cdot \Phi_\alpha$ or $\Gamma_\alpha \leq \frac{1040000 \cdot n}{\epsilon^8 \cdot \beta^4}$.
\end{lemma}

\begin{proof}
We use Lemma \ref{lem:psi_basis}.
In Eq. (\ref{eq:psi23}) we use $p_i > \frac{1+\epsilon}{n}$ for $i > 2n/3$ and $e^{-\alpha} - 1 < 0$.

\begin{eqnarray}
& & E[\Delta \Psi_\alpha \;|\; x(t)] \nonumber\\
& \leq & \sum_{i=1}^n e^{-\alpha y_i} \left( \beta(t) \left(p_i (e^{-\alpha} - 1) + \frac{\alpha}{n}\right) + (1-\beta(t)) \frac{1}{n} (e^\alpha-1-\alpha) + \frac{2\alpha^2}{n^2} \right) \nonumber\\
& \leq & \Psi_{\alpha, \leq 2n/3} \cdot \frac{\alpha}{n}\label{eq:psi23}\\
&&+ \sum_{i > 2n/3} e^{-\alpha y_i} \left( \beta(t) \left(\frac{1+\epsilon}{n} (e^{-\alpha} - 1) + \frac{\alpha}{n}\right) + (1-\beta(t)) \frac{1}{n} (e^\alpha-1-\alpha) + \frac{2\alpha^2}{n^2} \right) \nonumber\\
& \leq & \Psi_{\alpha, \leq \frac{2n}{3}} \cdot \frac{\alpha}{n} + \frac{\Psi_{\alpha, > \frac{2n}{3}}}{n} \left( \beta(t) \epsilon (e^{-\alpha} - 1) + \frac{\alpha^2 + \alpha^3}{2} \right) \nonumber\\
& \leq & \Psi_{\alpha, \leq \frac{2n}{3}} \cdot \frac{\alpha}{n} + \frac{\Psi_{\alpha, > \frac{2n}{3}}}{n} \left(-\alpha \beta(t) \epsilon + \alpha^2 \right) \nonumber\\
& \leq & \Psi_{\alpha, \leq \frac{2n}{3}} \cdot \frac{\alpha}{n} + \Psi_{\alpha, > \frac{2n}{3}} \cdot \frac{-\alpha \beta(t) \epsilon}{2n} \nonumber\\
& \leq & \Psi_{\alpha, \leq \frac{2n}{3}} \cdot \frac{\alpha}{n} + \Psi_{\alpha, > \frac{2n}{3}} \cdot \frac{-\alpha \beta \epsilon}{2n} \nonumber
\end{eqnarray}

Using $E[\Delta \Psi_\alpha \;|\; x(t)] \geq -\frac{\alpha \cdot \beta \cdot \epsilon}{4 \cdot n} \cdot \Psi_\alpha$, we get:

\begin{eqnarray*}
-\frac{\alpha \cdot \beta \cdot \epsilon}{4 \cdot n} \cdot \Psi_\alpha & \leq & \Psi_{\alpha, \leq \frac{2n}{3}} \cdot \frac{\alpha}{n} - \Psi_\alpha \cdot \frac{\alpha \beta \epsilon}{2n}\\
\Psi_\alpha & \leq & \frac{4}{\beta \cdot \epsilon} \cdot \Psi_{\alpha, \leq \frac{2n}{3}}
\end{eqnarray*}

If $\Psi_\alpha < \frac{\epsilon}{4} \cdot \Phi_\alpha$, we would be done; so we assume $\Psi_\alpha \geq \frac{\epsilon}{4} \cdot \Phi_\alpha$.
Since $\Psi_{\alpha, \leq \frac{2n}{3}} \leq \frac{2}{3} \cdot n \cdot e^{3 \cdot \alpha \cdot B / n}$ for $B = \sum_i \max\{0, y_i\}$ and since $y_{n/4} < 0$ implies $\Phi_\alpha \geq \frac{n}{4} \cdot e^{4 \cdot \alpha \cdot B / n}$, we get:

\begin{eqnarray*}
&& \frac{\epsilon \cdot n}{16} \cdot e^{4 \cdot \alpha \cdot B / n} \leq \frac{\epsilon}{4} \cdot \Phi_\alpha \leq \Psi_\alpha \leq \frac{4}{\beta \cdot \epsilon} \cdot \Psi_{\alpha, \leq \frac{2n}{3}} \leq \frac{8}{3 \cdot \beta \cdot \epsilon} \cdot n \cdot e^{3 \cdot \alpha \cdot B / n}\\
& \Rightarrow & e^{\alpha \cdot B / n} \leq \frac{128}{3 \cdot \beta \cdot \epsilon^2}
\end{eqnarray*}

Finally, it follows that:
\begin{eqnarray*}
\Gamma_\alpha & \leq & \frac{5}{\epsilon} \cdot \Psi_\alpha \leq \frac{5}{\epsilon} \cdot \frac{8}{3 \cdot \beta \cdot \epsilon} \cdot n \cdot e^{3 \cdot \alpha \cdot B / n} \leq \frac{40}{3 \cdot \beta \cdot \epsilon^2} \cdot n \cdot \left(\frac{128}{3 \cdot \beta \cdot \epsilon^2}\right)^3 \leq \frac{1040000 \cdot n}{\epsilon^8 \cdot \beta^4}
\end{eqnarray*}
\end{proof}

The next \cref{thm:recursion} utilizes all previous lemmas and corollaries to prove a bound on the potential change $\Delta \Gamma_\alpha = \Gamma_\alpha(x(t+1)) - \Gamma_\alpha(x(t))$.
This bound will then be used in \cref{thm:potential_bound1} to finally prove that the expected potential $\Exp[\Gamma_\alpha(x(t))]$ is linear in $n$.

\begin{lemma}[similar to Theorem A.9 in \cite{DBLP:conf/icalp/TalwarW14}] \label{thm:recursion}
$E[\Gamma_\alpha(t+1) \;|\; x(t)] \leq (1- \frac{\alpha \cdot \beta \cdot \epsilon}{4 \cdot n}) \cdot \Gamma_\alpha(t) + \frac{1040000}{\epsilon^8 \cdot \beta^4}$.
\end{lemma}

\begin{proof}
{\bf Case 1: $y_{n/4} \geq 0$ and $y_{3n/4} \leq 0$.}
The statement follows from Lemma~\ref{lem:phi1} and Lemma~\ref{lem:psi1}:
\begin{eqnarray*}
E[\Delta \Gamma_\alpha \;|\; x(t)] & = & E[\Delta \Phi_\alpha \;|\; x(t)] + E[\Delta \Psi_\alpha \;|\; x(t)]\\
& \leq & -\frac{\alpha \cdot \beta \cdot \epsilon}{n} \cdot \Phi_\alpha + \alpha - \frac{\alpha \cdot \beta \cdot \epsilon}{n} \cdot \Psi_\alpha + \alpha\\
& = & -\frac{\alpha \cdot \beta \cdot \epsilon}{n} \cdot \Gamma_\alpha + 2 \cdot \alpha
\end{eqnarray*}

{\bf Case 2: $y_{n/4} \geq y_{3n/4} > 0$.}
If $E[\Delta \Phi_\alpha \;|\; x(t)] \leq -\frac{\epsilon \cdot \beta \cdot \alpha}{4 \cdot n} \cdot \Phi_\alpha$, apply Lemma~\ref{lem:psi1}:
\begin{eqnarray*}
E[\Delta \Gamma_\alpha \;|\; x(t)] & = & E[\Delta \Phi_\alpha \;|\; x(t)] + E[\Delta \Psi_\alpha \;|\; x(t)]\\
& \leq & -\frac{\epsilon \cdot \beta \cdot \alpha}{4 \cdot n} \cdot \Phi_\alpha + \alpha - \Psi_\alpha \cdot \frac{\alpha \cdot \beta \cdot \epsilon}{n} + \alpha\\
& \leq & -\frac{\alpha \cdot \beta \cdot \epsilon}{4 \cdot n} \cdot \Gamma_\alpha + \alpha
\end{eqnarray*}

If this is not the case, apply Lemma~\ref{lem:phi2} which distinguishes two cases:
If $\Phi_\alpha < \frac{\epsilon}{4} \cdot \Psi_\alpha$, apply Lemma~\ref{lem:psi1} and Corollary~\ref{phi_corollary} and use $\Phi_\alpha < \frac{\epsilon}{4} \cdot \Psi_\alpha \leq \frac{1}{16} \cdot \Psi_\alpha$ as well as $\beta \geq \alpha$ (which follows from $\epsilon \geq 2 \cdot \alpha / \beta$):
\begin{eqnarray*}
E[\Delta \Gamma_\alpha \;|\; x(t)] & = & E[\Delta \Phi_\alpha \;|\; x(t)] + E[\Delta \Psi_\alpha \;|\; x(t)]\\
& \leq & \frac{\alpha^2}{n} \cdot \Phi_\alpha - \Psi_\alpha \cdot \frac{\alpha \cdot \beta \cdot \epsilon}{n} + \alpha\\
& \leq & \frac{\alpha \cdot \beta \cdot \epsilon}{n \cdot 4} \cdot \Psi_\alpha - \Psi_\alpha \cdot \frac{\alpha \cdot \beta \cdot \epsilon}{n} + \alpha\\
& = & -\Psi_\alpha \cdot \frac{3 \cdot \alpha \cdot \beta \cdot \epsilon}{4 \cdot n} + \alpha\\
& \leq & -\Gamma_\alpha \cdot \frac{3 \cdot \alpha \cdot \beta \cdot \epsilon}{4 \cdot n} + \alpha
\end{eqnarray*}

If $\Gamma_\alpha < \frac{360 \cdot n}{\epsilon^8 \cdot \beta^4}$, apply Corollary~\ref{phi_corollary} and Corollary~\ref{psi_corollary}:
\begin{eqnarray}
E[\Delta \Gamma_\alpha \;|\; x(t)] & = & E[\Delta \Phi_\alpha \;|\; x(t)] + E[\Delta \Psi_\alpha \;|\; x(t)] \label{eq:cn_start}\\
& \leq & \frac{\alpha^2}{n} \cdot \Phi_\alpha + \frac{\alpha^2}{n}  \cdot \Psi_\alpha \nonumber\\
& = & \frac{\alpha^2}{n} \cdot \Gamma_\alpha \nonumber\\
& \leq & -\frac{\alpha \cdot \beta \cdot \epsilon}{n} \cdot \Gamma_\alpha + \frac{(1+\epsilon) \cdot \alpha \cdot \beta}{n}\cdot \Gamma_\alpha \nonumber\\
& \leq & -\frac{\alpha \cdot \beta \cdot \epsilon}{n} \cdot \Gamma_\alpha + \frac{360}{\epsilon^8 \cdot \beta^2} \label{eq:cn_end}
\end{eqnarray}

{\bf Case 3: $y_{3n/4} \leq y_{n/4} < 0$.}
If $E[\Delta \Psi_\alpha \;|\; x(t)] \leq -\frac{\alpha \cdot \beta \cdot \epsilon}{4 \cdot n} \cdot \Psi_\alpha$, apply Lemma~\ref{lem:phi1}:
\begin{eqnarray*}
E[\Delta \Gamma_\alpha \;|\; x(t)] & = & E[\Delta \Phi_\alpha \;|\; x(t)] + E[\Delta \Psi_\alpha \;|\; x(t)]\\
& \leq & -\frac{\alpha \cdot \beta \cdot \epsilon}{n} \cdot \Phi_\alpha + \alpha - \frac{\alpha \cdot \beta \cdot \epsilon}{4 \cdot n} \cdot \Psi_\alpha\\
& \leq & -\frac{\alpha \cdot \beta \cdot \epsilon}{4 \cdot n} \cdot \Gamma_\alpha + \alpha
\end{eqnarray*}

If this is not the case, apply Lemma~\ref{lem:psi2} which distinguishes two cases:
If $\Psi_\alpha < \frac{\epsilon}{4} \cdot \Phi_\alpha$, apply Lemma~\ref{lem:phi1} and Corollary~\ref{psi_corollary} and use $\epsilon < 1/4$:
\begin{eqnarray*}
E[\Delta \Gamma_\alpha \;|\; x(t)] & = & E[\Delta \Phi_\alpha \;|\; x(t)] + E[\Delta \Psi_\alpha \;|\; x(t)]\\
& \leq & -\frac{\alpha \cdot \beta \cdot \epsilon}{n} \cdot \Phi_\alpha + \alpha + \frac{\alpha^2}{n} \cdot \Psi_\alpha\\
& \leq & -\frac{3 \cdot \alpha \cdot \beta \cdot \epsilon}{4 \cdot n} \cdot \Gamma_\alpha + \alpha
\end{eqnarray*}
(The proof for the last inequality is very similar to Case 2.)

If $\Gamma_\alpha < \frac{c \cdot n}{\beta^4}$, Inequality~(\ref{eq:cn_start})-(\ref{eq:cn_end}) can be reused with the constant from Lemma~\ref{lem:psi2}.
\begin{eqnarray*}
E[\Delta \Gamma_\alpha \;|\; x(t)] & \leq & -\frac{\alpha \cdot \beta \cdot \epsilon}{n} \cdot \Gamma_\alpha + \frac{1040000}{\epsilon^8 \cdot \beta^2}
\end{eqnarray*}
\end{proof}

\begin{lemma} \label{thm:potential_bound1}
Fix an arbitrary step $t$. Let $\epsilon \leq \frac{3}{16}$, and assume that $\beta \in (0, 1]$ is a constant probability and that $\alpha \leq \frac{\epsilon \cdot \beta}{3}$.
Then
\begin{equation*}
\Exp[\Gamma_\alpha(x(t))]
\leq
\frac{4160002}{\alpha \cdot \beta^3 \cdot \epsilon^9} \cdot e^{\alpha \cdot \adisc(0)} \cdot n
.
\end{equation*}
\end{lemma}
\begin{proof}
The proof of this lemma follows from Lemma~\ref{thm:recursion} which gives us
\begin{equation*}
\Exp[\Gamma_\alpha(x(t+1)) \;|\; x(t)]
\leq
\left(1 - \frac{\alpha \cdot \beta \cdot \epsilon}{4 \cdot n}\right) \cdot \Gamma_\alpha(x(t)) + \frac{1040000}{\epsilon^8 \cdot \beta^2}
.
\end{equation*}
Taking the expected value on both sides, we obtain
\begin{equation*}
\Exp[\Gamma_\alpha(x(t+1))]
\leq
\left(1 - \frac{\alpha \cdot \beta \cdot \epsilon}{4 \cdot n}\right) \cdot \Exp[\Gamma_\alpha(x(t))] + \frac{1040000}{\epsilon^8 \cdot \beta^2}.
\end{equation*}
Solving this recursion yields
\begin{equation*}
\begin{aligned}
\Exp[\Gamma_\alpha(x(t))]
&\leq
{\left(1 - \frac{\alpha \cdot \beta \cdot \epsilon}{4 \cdot n}\right)}^t \cdot \left(\Exp[\Gamma_\alpha(x(0))] - \frac{4160000}{\alpha \cdot \beta^3 \cdot \epsilon^9} \cdot n\right) + \frac{4160000}{\alpha \cdot \beta^3 \cdot \epsilon^9} \cdot n
\\&\leq
\Exp[\Gamma_\alpha(x(0))] + \frac{4160000}{\alpha \cdot \beta^3 \cdot \epsilon^9} \cdot n
\\&\leq
2n \cdot e^{\alpha \cdot \adisc(0)} + \frac{4160000}{\alpha \cdot \beta^3 \cdot \epsilon^9} \cdot n
\\&=
\left( 2e^{\alpha \cdot \adisc(0)} + \frac{4160000}{\alpha \cdot \beta^3 \cdot \epsilon^9} \right) \cdot n
\leq
\frac{4160002}{\alpha \cdot \beta^3 \cdot \epsilon^9} \cdot e^{\alpha \cdot \adisc(0)} \cdot n
.
\end{aligned}
\end{equation*}
\end{proof}

\begin{lemma}\label{lem:whppotential_bound1}
Fix an arbitrary step $t$.
Let $\epsilon \leq \frac{3}{16}$, and assume that $\beta \in (0, 1]$ is a constant probability and that $\alpha \leq \frac{\epsilon \cdot \beta}{3}$.
Then
\begin{equation*}
\Pr{\adisc(t) - \adisc(0) \geq \frac{(a + 2) \cdot \ln n}{\alpha}}
\leq
n^{-a}
.
\end{equation*}
\end{lemma}
\begin{proof}
\Cref{thm:potential_bound1} together with Markov's inequality yields
\begin{equation*}
\Pr{\Gamma_{\alpha}(x(t)) \geq n^{a+1} \cdot \frac{4160002}{\alpha \cdot \beta^3 \cdot \epsilon^9} \cdot e^{\alpha \cdot \adisc(0)}}
\leq
n^{-a}
.
\end{equation*}
By rearranging the inequality on the left-hand side, we see that the corresponding event is equivalent to
\begin{equation*}
\frac{\ln(\Gamma_{\alpha(x(t))})}{\alpha} - \adisc(0)
\geq
\frac{(a+1) \cdot \ln(n) + \ln(\frac{4160002}{\alpha \cdot \beta^3 \cdot \epsilon^9})}{\alpha}
.
\end{equation*}
For large $n$, the right-hand side of this is at most
\begin{math}
\frac{(a+2) \cdot \ln(n)}{\alpha}
\end{math}.
For the left-hand side, the potential function definition yields a lower bound of
\begin{math}
\adisc(t) - \adisc(0)
\end{math}.
Combining the above insights, we get
\begin{equation*}
\begin{aligned}
\Pr{\adisc(t) - \adisc(0) \geq \frac{(a + 2) \cdot \ln n}{\alpha}}
&\leq
\Pr{\frac{\ln(\Gamma_{\alpha(x(t))})}{\alpha} - \adisc(0) \geq \frac{(a+1) \cdot \ln(n) + \ln(\frac{4160002}{\alpha \cdot \beta^3 \cdot \epsilon^9})}{\alpha}}
\\&=
\Pr{\Gamma_{\alpha}(x(t)) \geq n^{a+1} \cdot \frac{4160002}{\alpha \cdot \beta^3 \cdot \epsilon^9} \cdot e^{\alpha \cdot \adisc(0)}}
\leq
n^{-a}
.
\end{aligned}
\end{equation*}
\end{proof}

    \section{Omitted Material from Proof of Theorem \ref{prop:linear_potential}.}
\label{app:tom}

We restate \cref{prop:linear_potential} for convenience.

\restateThmTom*

\begin{proof}
From \cref{lem:adaption_sauerwald} we may assume that $\Gamma_\alpha(x(t)) = O(n)$ with probability at least $1-n^{-(3+\ell)}$ for any fixed non-negative $\ell$, where the constant factor hidden in the big-Oh depends on $\ell$.

We assign a height $h_b(t)$ to each ball $b$ at time $t$ in such a way that in each bin $i$ the $X_i(t)$ balls contained in that bin at time $t$ are numbered $1,\ldots,X_i(t)$. Further let $h'_b(t)=h_b(t)-\avg$ if $h_b(t)\ge\avg$ and $h'_b(t)=0$ otherwise.

The standard potential function $\Phi_\alpha(t)=\sum_i e^{\alpha X^+_i(t)}$ with $X^+_i(t)$ being bin $i$'s load above \avg at time $t$ and zero otherwise is expressed in terms of individual bins' potentials.  We will rearrange this function so as to express the potential in terms of individual balls' potentials, and then make subsequent arguments based upon this new potential function.

We assign potential $\pi_b(t)$ to ball $b$ at height $h_b(t)$ as follows:
\[
  \pi_b(t) = \left\{
    \begin{array}{ll}
      0 & \mbox{if $h_b(t)\le \avg$}, \\
      e^{\alpha} & \mbox{if $h_b(t)=\avg+1$}, \\
      e^{\alpha h'_b(t)} - e^{\alpha (h'_b(t)-1)} &
        \mbox{if $h_b(t)\ge \avg+2$}.
    \end{array}
  \right.
\]
Furthermore let $\Pi(t) = \sum_{\text{ball}\, b} \pi_b(t)$. Now $\Pi(t)=\Phi(t)$, which may be seen as follows: Consider a bin $i$ with $k\ge 1$ balls above $\avg$.  Then those balls contribute a total potential of
\begin{align*}
\Pi_i(t) &:= \sum_{\text{ball}\, b:\ h_b(t)>\avg} \pi_b(t) =
e^\alpha + \sum_{j=2}^k e^{\alpha\cdot j} - e^{\alpha \cdot(j-1)} \\
&= e^\alpha + (e^{2\alpha}-e^\alpha) + (e^{3\alpha}-e^{2\alpha}) + \cdots +
(e^{(k-1)\alpha}-e^{(k-2)\alpha}) +
(e^{k\alpha}-e^{(k-1)\alpha}).
\end{align*}
This is a telescoping sum evaluating to $\Pi_i(t) = e^{k\alpha}$. This however is precisely this bin's potential contribution in the potential function $\Phi$, so that we can conclude that
\[
\Pi(t) =
\sum_{\text{ball}\, b} \pi_b(t) =
\sum_{\text{bin}\, i} \ \sum_{\substack{\text{ball}\, b\\\text{in bin}\, i}} \pi_b(t) =
\sum_{\text{bin}\, i} \ \Pi_i(t) =
\sum_{\text{bin}\, i} e^{\alpha X_i^+(t)} = \Phi(t).
\]%

Consider a height $\avg+k$ for some $k\ge 2$; the precise value will be fixed below.
Assume that there is a set ${\cal B}$ of at least $\mu n$ balls at height $\avg+k$ or higher (if no such set exists the proof is finished); the precise distribution of the balls in bins is irrelevant.
Each of those at least $\mu n$ many balls contributes a potential of at least
\[
e^{\alpha k} - e^{\alpha (k-1)} =
e^{\alpha k} \cdot\left(1-\frac{1}{e^\alpha}\right) \ge
\epsilon \cdot e^{\alpha k}
\]
for some constant $\epsilon=\epsilon(\alpha)>0$, for a total potential contribution of those balls of at least
\[
\Pi_{\cal B}(t) \ge
\epsilon\cdot \mu n\cdot e^{\alpha k}.
\]
Now $\Pi_{\cal B}(t) > c\cdot n$ whenever $k>\frac{1}{\alpha}\ln\frac{c}{\epsilon\mu}-1$, so letting $d=\frac{1}{\alpha}\ln\frac{c}{\epsilon\mu}$ proves the theorem (recall that the probability is inherited from \cref{lem:adaption_sauerwald} -- the actual argument hee in the proof building on top of it is entirely deterministic).
\end{proof}

\begin{lemma}\label{lem:adaption_sauerwald}
    Assume the premise of \cref{prop:linear_potential} holds.
    Then $\Gamma_\alpha(x(t)) = O(n)$ with probability at least $1-n^{-(3+\ell)}$ for any fixed non-negative $\ell$, where the constant factor hidden in the big-Oh depends on $\ell$.
\end{lemma}

\begin{proof}
The proof is a minor adaptation of Theorem 5.3 on page 18 in Los and Sauerwald's \cite{LS21}, which is a more detailed version in arXiv of their \cite{DBLP:conf/innovations/LosS22}.  All references in \cref{app:tom} to their lemmas etc.~are with respect to the numbering in their arXiv version \cite{LS21}.

As mention before, the  purpose of the adaptation of \cite{LS21}'s results is two-fold: (i) to incorporate our deletions and their respective probability bounds, and (ii) to improve the high-probability guarantee from the original $1-1/n^3$ (which was sufficient for the purposes of \cite{DBLP:conf/innovations/LosS22, LS21})  to $1-1/n^{3+\ell}$ for any arbitrary but fixed $\ell>0$.
\medskip

To show this result we adapt Theorem 5.3 in \cite{LS21} as follows.  The $\ell$ in the exponent of the statement comes from \cref{prop:linear_potential}.

\begin{theorem}[cf. Theorem 5.3 in \cite{LS21}] \label{thm:linear_potential}
Consider any probability vector $p$ that is (i) non-decreasing in $i$, i.e., $p_i \leq p_{i+1}$ and (ii) for constant $\epsilon$, $\sum_{i \geq 3n/4} p_i \geq 1/4+\epsilon$ and $\sum_{i \leq n/4} p_i \geq 1/4-\epsilon$.
Then, for any $t \geq 0$ and $\alpha_2 \leq \frac{\beta \cdot \epsilon}{110}$, $c = c_{\epsilon, \alpha_2} = 2 \cdot 40 \cdot 128^3 \cdot \epsilon^{-7} \cdot 4 \cdot \alpha_2^{-1}$,
$$
\Pr{\cap_{s \in [t, t+n \log^5(n)]} \{\Gamma_{\alpha_2}^{(s)} \leq 2cn\}} \geq 1 - n^{-(3+\ell)}.
$$
\end{theorem}

\begin{proof}
The proof is almost identical to that of Theorem 5.3 in \cite{LS21}.  The chain of modifications starts with Lemma 5.5 in \cite{LS21} where the restriction $\Gamma_1^{(t)} \le c n^9$ occurs.  An application of Markov's inequality turns the $n^9$ into a probability $1/n^8$, which subsequent lemmas etc.~further decay to $1/n^3$. The main idea to make the result work for us is to modify the restriction on $\Gamma_1^{(t)}$ to $\Gamma_1^{(t)} \le c n^{9+\ell}$, which accordingly will turn into a probability $1/n^{8+\ell}$, only to subsequently  be eroded to $1/n^{3+\ell}$. Basically we have the exact same results with the exact same proofs as in in \cite{LS21} but with some explicit constant parameters replaced with their respective original values plus $\ell$. We will need to make some equally straightforward modification of \cite{LS21}'s results to make them compatible with our potential function analysis.

\smallskip

We will use two instances of the potential function $\Gamma_\alpha$ which only differ in $\alpha$.
The first one, $\Gamma_{\alpha_1}$, uses $\alpha_1 \leq \frac{\beta \cdot \epsilon}{3}$, the second one, $\Gamma_{\alpha_2}$, uses $\alpha_2 \leq \frac{\alpha_1}{4 \cdot 9.1}$.
(Explanation: $\alpha_1$ and $\alpha_2$ must fulfil the requirements of the Potential Function Analysis in the previous section and of Lemma~\ref{lem:ls22}.)
From \cref{thm:potential_bound1} we can derive, for $n$ and $c$ large enough, i.e. $c \geq \frac{4160002 \cdot n}{\alpha \cdot \beta^3 \cdot \epsilon^9}$, that $\Pr{\Gamma_{\alpha_1}(x(t)) \geq c \cdot n^{9+\ell}} < n^{-(8+\ell)}$.
The next lemma  proves a stronger statement for $\Gamma_{\alpha_2}$.

\begin{lemma}[cf. Lemma 5.5 in \cite{LS21}] \label{lem:ls22}
For any $t \geq 0$, if $\Gamma_{\alpha_1}^{(t)} \leq c \cdot n^{9+\ell}$, then, (i) $|x_i(t)| \leq \frac{9.1}{\alpha_1} \cdot \log(n)$ for all $i \in [n]$, (ii) $\Gamma_{\alpha_2}^{(t)} \leq n^{4/3}$, and (iii) $|\Gamma_{\alpha_2}^{(t+1)} - \Gamma_{\alpha_2}^{(t)}| \leq n^{1/3}$.
\end{lemma}

\begin{proof}
First statement: For all $i \in [n]$,
\[
\Gamma_{\alpha_1}^{(t)} = e^{\alpha_1 \cdot x_i(t)} + e^{-\alpha_1 \cdot x_i(t)} \leq c \cdot n^{9+\ell}
\implies  |x_i(t)| \leq \frac{9.1}{\alpha_1} \cdot \log(n).
\]

Second statement:
$$
\Gamma_{\alpha_2}^{(t)} < 2 \cdot \sum_{i=1}^n \exp\left(\alpha_2 \cdot \frac{9.1}{\alpha_1} \cdot \log(n)\right) \leq 2 \cdot n \cdot n^{1/4} < n^{4/3}.
$$

Third statement:
One of the $x_i(t)$ is increased or decreased by $1$, and all $x_i(t)$ are (additionally) decreased or increased by $1/n$ as the average load is altered by the allocation or deletion in step $t+1$.
The change affects both parts, $\Phi_\alpha$ and $\Psi_\alpha$, of potential function $\Gamma_\alpha$.
In total, the absolute $\ell_1$-change in all $x_i(t)$ is upper-bounded by $4$.
Since $exp(.)$ is convex, the hypothetical worst case is that the largest exponent increases by $4$ while all others remain the same:
\begin{eqnarray*}
|\Gamma_{\alpha_2}^{(t+1)} - \Gamma_{\alpha_2}^{(t)}| & \leq & \exp(\alpha_2 \cdot \max\{x_{\max}^{(t)} + 4, -x_{min}^{(t)} - 4\})\\
& \leq & e^{4 \alpha_2} \cdot \exp(\alpha_2 \cdot \frac{9.1}{\alpha_1} \cdot \log(n)) \leq n^{1/3}\\
\end{eqnarray*}
\end{proof}


If we define 
$c' = \frac{104 \cdot 10^6}{\epsilon^8 \cdot \beta^4}$
and $\epsilon'_\alpha = \frac{\alpha \beta \epsilon}{4}$ then we obtain the following claim by applying Theorem~\ref{thm:recursion} in the proof from \cite{LS21}.

\begin{claim}[cf. Claim 5.7 in \cite{LS21}]
For any $t \geq 0$, $E[\Gamma_{\alpha_2}^{(t+1)} \,|\, \Gamma_{\alpha_2}^{(t)}, \Gamma_{\alpha_2}^{(t)} \geq \frac{2c'}{\epsilon'_{\alpha_2}} \cdot n] \leq (1 - \frac{\epsilon'_{\alpha_2}}{2n}) \cdot \Gamma_{\alpha_2}^{(t)}.$
\end{claim}

\begin{proof}
  The proof is identical to Claim 5.7 in \cite{LS21}.
\end{proof}

\begin{lemma}[cf. Lemma 5.6 in \cite{LS21}]\label{lem:our56}
For any $t \geq n \cdot \log^2(n)$, for constants $c' > 0$, $\epsilon'_{\alpha_2} > 0$ defined as above,
$$
\Pr{\cup_{s \in [t-n \cdot \log^2(n), t]} \left\{ \Gamma_{\alpha_2}^{(s)} \leq \frac{2c'}{\epsilon'_{\alpha_2}} \cdot n \right\}} \geq 1 - 2 \cdot c \cdot n^{-(8+\ell)}.
$$
\end{lemma}

\begin{proof}
    The proof is almost identical to that of Lemma 5.6 in \cite{LS21}.  In fact the modifications immediately follow from those made to Lemma 5.5 (here: Lemma \ref{lem:ls22}).
\end{proof}

We are now in a position to finish the proof of Theorem \ref{thm:linear_potential}, a part that essentially mimics Section 5.2, with just some few slightly adapted coefficients. We do of course not wish to copy a page or two from another paper, and therefore just briefly highlight the modifications.

The application of Lemma 5.6 in \cite{LS21} (here Lemma \ref{lem:our56}) gives a lower bound of $1-2cn^{-(8+\ell)}$ instead of $1-2cn^{-8}$ for the probability.  Further down, the application of the union bound gives
\[
\Pr{ \bigcap_{u \in [t-n \log^2 n, t+ n \log^ 5 n]} \{ \Gamma_1^{(u)} \leq c n^{9+\ell} \} } \geq 1 - \frac{2\log^5 n}{n^{7+\ell}}.
\]
Then, the application of Theorem A.4 in \cite{LS21} gives
\[
\Pr{X_r^{(u)} \geq X_r^{(r)} + \frac{cn}{2}} \leq \exp\left( - \frac{c^2n^2/4}{2 \cdot (2n \log^5 n) \cdot (2 n^{1/3}) } \right) + \frac{2\log^5 n}{n^{7+\ell}} \leq  \frac{3 \log^5 n}{n^{7+\ell}}.
\]
Further down we find that if a red phase starts at time $r$, then with probability $1-\frac{3 \log^5 n}{n^{7+\ell}}$, $\Gamma_{2}^{(u)}$ will always be $\leq \frac{3cn}{2}$. Taking union bounds as in Section 5.2 in \cite{LS21} yields
\[
\Pr{\bigcup_{r \in [s, t + n \log^5 n]}  \bigcup_{u \in [r, t + n \log^5 n]} \{ X_{r}^{(u)} > \frac{3cn}{2} \} } \leq 3 \cdot \frac{\log^5 n}{n^{7+\ell}} \cdot (4 n^2 \log^{10} n) \leq \frac{1}{2} n^{-{4+\ell}}.
\]
Concluding as in \cite{LS21}, with probability $1-\frac{1}{2}n^{-{4+\ell}}-2cn^{-{8+\ell}} \geq 1-n^{-{4+\ell}}$, it holds that $\Gamma_2^{(r)} \leq \frac{3cn}{2}$ for all time-steps $r$ which are within a red phase in $[s,t+n \log^ 5 n] \subseteq [t,t+n \log^5 n]$.
Since $\Gamma_2^{(r)} \leq T \leq cn$ holds (deterministically) by definition for all time-steps $r$ within a green phase, the theorem follows.
\end{proof}
This concludes the proof sketch of the Lemma.
\end{proof}

    \section{Omitted Material from Section \ref{sec:layered}.}
\label{app:layered}

We first present the proof of \cref{lem:murw_drift}.
We restate it for convenience:
\restatemurwdrift*
\begin{proof}
For the sake of readability, we abuse notation in this proof by omitting the base height $h$ when talking about the number of balls that have at least some level $\ell$ (i.e., we write $m_{\ell}(\bullet)$ instead of $m_{\ell}^{(h)}(\bullet)$).
Similarly, we refer to the notations $h$-valid and $h$-critical simply as valid/critical within the proof.
\begin{enumerate}
\item
    Consider a critical level $\ell \in \set{1, 2, \dots, \ell^*}$ of configuration $\bm{x}$.
    Since $\bm{x}$ is valid, level $\ell - 1$ and $\ell + 1$ are valid.
    Using  \cref{lem:muchangeprob} and \cref{obs:lvlboundary:props} we calculate
    \begin{equation*}
    \begin{aligned}
    \frac{\Pr{\nrballslvlgeq{\ell}{\bm{x'}} < \nrballslvlgeq{\ell}{\bm{x}}}}
         {\Pr{\nrballslvlgeq{\ell}{\bm{x'}} > \nrballslvlgeq{\ell}{\bm{x}}}}
    &\geq
    \frac{1 - \beta(t)}{\beta(t)} \cdot \frac{{\nrballslvlgeq{\ell}{\bm{x}} - \nrballslvlgeq{\ell + 1}{\bm{x}}}}{{\bigl( \nrballslvlgeq{\ell - 1}{\bm{x}} - \nrballslvlgeq{\ell}{\bm{x}} \bigr)}^2} \cdot n
    \\&\geq
    \frac{1 - \beta}{\beta} \cdot \frac{{\alpha_{\ell}/2 - \alpha_{\ell + 1}}}{{(\alpha_{\ell - 1} - \alpha_{\ell}/2)}^2} \cdot n
    \geq
    \frac{1 - \beta}{\beta} \cdot \frac{\alpha_{\ell}/4}{\alpha_{\ell - 1}^2} \cdot n
    \geq
    2
    .
    \end{aligned}
    \end{equation*}
    The first inequality follows from \cref{lem:muchangeprob}.
    The second inequality uses the bound on $\beta(t)$, that level $\ell$ is critical, and that levels $\ell - 1$ and $\ell + 1$ are valid.
    The third and final inequalities follow from \cref{obs:lvlboundary:props}.

\item For level $\ell^* + 1$ we similarly calculate
    \begin{equation*}
    \begin{aligned}
    \frac{\Pr{\nrballslvlgeq{\ell^* + 1}{\bm{x'}} < \nrballslvlgeq{\ell^* + 1}{\bm{x}}}}
         {\Pr{\nrballslvlgeq{\ell^* + 1}{\bm{x'}} > \nrballslvlgeq{\ell^* + 1}{\bm{x}}}}
    &\geq
    \frac{1 - \beta(t)}{\beta(t)} \cdot \frac{{\nrballslvlgeq{\ell^* + 1}{\bm{x}} - \nrballslvlgeq{\ell^* + 1 + 1}{\bm{x}}}}{{\bigl( \nrballslvlgeq{\ell^* + 1 - 1}{\bm{x}} - \nrballslvlgeq{\ell^* + 1}{\bm{x}} \bigr)}^2} \cdot n
    \\&\geq
    \frac{1 - \beta}{\beta} \cdot \frac{1}{\alpha_{\ell^*}^2} \cdot n
    =
    \frac{1 - \beta}{\beta} \cdot \frac{n}{144\log^2(n)}
    =
    \Omega\left(\frac{n}{\log^2(n)}\right)
    \geq
    \sqrt{n}
    .
    \end{aligned}
    \end{equation*}
    Compared to the proof of the previous statement, we bound the number of balls at level at least $\ell^* + 1$ in the numerator by $1$ (level $\ell^* + 1$ being critical means it is non-empty, so there is at least one bin with a ball on level $\ell^* + 1$).
    The rest of the calculation applies the definition of $\alpha_{\ell^*}$ (see \cref{eqn:lvlboundary:def}) and simplifies.
    \qedhere
\end{enumerate}
\end{proof}

We next present the proof of \cref{lem:maintainvalidity}.
We restate it for convenience:
\restatemaintainvalidity*
\begin{proof}
\newcommand*{\eCross}{\cC}
Define the stopping time
\begin{math}
\tau
\coloneqq
\min(\set{T} \cup \set{t \in \N | \text{$\bm{x(t)}$ is not $h(t)$-valid}})
\end{math}.
With this, we clearly have
\begin{math}
\Pr{\forall t \in \set{0, 1, \dots, T - 1}\colon\text{$\bm{x(t)}$ is $h(t)$-valid}}
=
\Pr{\tau \geq T}
\end{math},
so it is sufficient to prove that
\begin{math}
\Pr{\tau \geq T}
\geq
1 - n^{-1}
\end{math}.
In the following, we say \emph{level $\ell$ is valid at time $t$} (critical/safe) if level $\ell$ of configuration $\bm{x(t)}$ is $h(t)$-valid (-critical/-safe).
We use the shorthand $\tilde{m}_{\ell}(t) \coloneqq m_{\ell}^{(h(t))}(\bm{x(t)})$.

Since the base height $h(t)$ is monotonously non-decreasing in $t$, we have
\begin{math}
\tilde{m}_{\ell}(t+1)
=
m_{\ell}^{h(t+1)}(\bm{x(t+1)})
\leq
m_{\ell}^{h(t)}(\bm{x(t+1)})
\end{math}
for any $t \in \N_0$.
Thus, if $\bm{x(t)}$ is $h(t)$-valid and level $\ell$ is critical at time $t$, \cref{lem:murw_drift} implies
\begin{equation}%
\label{eqn:maintainvalidity:murw_drift}
\begin{aligned}
\frac{\Pr{\tilde{m}_{\ell}(t+1) < \tilde{m}_{\ell}(t)}}
     {\Pr{\tilde{m}_{\ell}(t+1) > \tilde{m}_{\ell}(t)}}
&\geq
\frac{\Pr{m_{\ell}^{(h(t))}(\bm{x(t+1)}) < m_{\ell}^{(h(t))}(\bm{x(t)})}}
     {\Pr{m_{\ell}^{(h(t))}(\bm{x(t+1)}) > m_{\ell}^{(h(t))}(\bm{x(t)})}}
\\&\geq
\begin{cases}
2        & \text{if $\ell \in \set{0, 1, \dots, \ell^*}$ and}\\
\sqrt{n} & \text{if $\ell = \ell^* + 1$.}
\end{cases}
\end{aligned}
\end{equation}

If $\tau < T$, there must be some level $\ell$ that is not valid at time $\tau$ and a \emph{minimal} time $t' < \tau$ such that level $\ell$ is critical at \emph{any} integral time $t \in \intco{t', \tau}$.
In particular, using the minimality of $t'$, we have for all integral $t \in \intco{t', \tau}$,
$\tilde{m}_{\ell}(t')
=
\alpha_{\ell} / 2
\leq
\tilde{m}_{\ell}(t)
<
\tilde{m}_{\ell}(\tau)
=
\alpha_{\ell}
$.
That is, between time step $t'$ and $\tau$, the number of balls at level at least $\ell$ started at $\alpha_{\ell}/2$ and increased to $\alpha_{\ell}$ \emph{without} decreasing below $\alpha_{\ell}/2$.
If this happens, we say level $\ell$ \emph{crosses} $\intco{\alpha_{\ell} / 2, \alpha_{\ell}}$ during $\intco{t', \tau}$, and denote this event as $\eCross_{\ell, \intco{t', \tau}}$.
We will show that for \emph{any} level $\ell$ and \emph{any} time $t'$ event $\eCross_{\ell, \intco{t', \tau}}$ is unlikely to happen.
Using a union bound, this yields that, \whp, $\tau \geq T$, as desired.
To bound $\Pr{\eCross_{\ell, \intco{t', \tau}}}$, we define a random walk $W^{(\ell, t')}$ that has an even larger probability to cross $\intco{\alpha_{\ell} / 2, \alpha_{\ell}}$ during $\intco{t', \tau}$ and show that this larger probability is still small enough.

To this end, fix a level $\ell \in \set{0, 1, \dots, \ell^* + 1}$ and integral time $t' \in \intcc{0, \tau} \subseteq \intcc{0, T}$.
Let $p^{(\ell, t)} \coloneqq \Pr{\tilde{m}_{\ell}(t + 1) = \tilde{m}_{\ell}(t) \;|\; \bm{x(t)}}$ denote the probability that the number of balls at level at least $\ell$ does not change from time $t$ to time $t + 1$ (note that this probability is a random variable depending on the random configuration $\bm{x(t)}$).
We define the random walk $W^{(\ell, t')} = \intoo{W^{(\ell, t')}_t}_{t \in \intco{t', T}}$ over the discrete time interval $\intco{t', T}$.
The random walk's start position is
\begin{math}
W^{(\ell, t')}_{t'}
\coloneqq
\tilde{m}_{\ell}(t')
=
\alpha_{\ell} / 2
\end{math}
and its movement from time $t \in \intco{t', T}$ to time $t + 1$ is defined via
\begin{equation}
W^{(\ell, t')}_{t + 1} = \begin{cases}
    W^{(\ell, t')}_t     & \text{with probability $p^{(\ell, t)}$,}\\
    W^{(\ell, t')}_t - 1 & \text{with probability $\left(1 - p^{(\ell, t)}\right) \cdot 2/3$, and}\\
    W^{(\ell, t')}_t + 1 & \text{with probability $\left(1 - p^{(\ell, t)}\right) \cdot 1/3$.}
\end{cases}
\end{equation}
In other words, the random walk $W^{(\ell, t')}$
\begin{enumerate*}[afterlabel=, label=]
\item starts at time $t'$ at position $\tilde{m}_{\ell}(t') = \alpha_{\ell} / 2$,
\item maintains its position whenever $\tilde{m}_{\ell}(\bullet)$ does not change, and
\item the probability that it moves left is twice the probability that it moves right.
\end{enumerate*}
For a time $t \geq t'$, we say $W^{(\ell, t')}$ \emph{crosses} $\intco{\alpha_{\ell} / 2, \alpha_{\ell}}$ during $\intco{t', t}$ if $W^{(\ell, t')}$ reaches $\alpha_{\ell}$ before decreasing below $\alpha_{\ell}/2$, and denote this event as $\eCross^W_{\ell, \intco{t', t}}$.

Given the current configuration $\bm{x(t)} = \bm{x}$ at time $t \geq t'$ and depending on how $\tilde{m}_{\ell}(\bullet)$ changes from time $t$ to time $t + 1$, we can easily couple the random walk $W^{(\ell, t')}$ with our process as follows:
\begin{itemize}
\item if $\tilde{m}_{\ell}(\bullet)$ does not change, $W^{(\ell, t')}$ remains stationary;
\item if $\tilde{m}_{\ell}(\bullet)$ increases, $W^{(\ell, t')}$ moves right; and
\item if $\tilde{m}_{\ell}(\bullet)$ decreases, $W^{(\ell, t')}$ moves left with prob.~$\min\set{1, \frac{2}{3} / \Pr{\tilde{m}_{\ell}(t + 1) < \tilde{m}_{\ell}(t) \;\mid\; \bm{x(t)} = \bm{x}}}$ and right otherwise.
\end{itemize}
By this coupling, the random walk is never \enquote{left of} $\tilde{m}_{\ell}(\bullet)$ (i.e., for all $t \geq t'$ we have $W^{(\ell, t')}_t \geq \tilde{m}_{\ell}(t)$).

We now consider the crossing probabilities for the different levels $\ell \in \set{0, 1, \dots, \ell^* + 1}$ for some integral time $t' \in \intco{0, \tau}$ at which level $0$ just became critical.
For level $\ell = 0$, \cref{prop:linear_potential} can be used to see that, w.h.p, this level will not become invalid at any time during the discrete interval $\intco{0, \tau}$.
For any other level $\ell$, we use the fact that the $\tilde{m}_{\ell}$-value is upper-bounded by the position of the random walk $W^{(\ell, t')}$ during $\intco{t', \tau}$ (by the coupling above) together with the fact that the left-biased random walk $W^{(\ell, t')}$ is very unlikely to reach the right border before becoming again safe (by \cref{lem:biasedrw_crossprob}).
\begin{description}
\item[Case~1:]
    For level $\ell = 0$ we do not need to considered the random walk.
    Instead, note that \cref{prop:linear_potential} allows us to choose a constant $\gamma = \gamma(\alpha_0)$ such that the number of balls above height $\ceil{m(t)/n} + \gamma$ is, with (arbitrarily) high probability, at most $\alpha_0 = \frac{1 - \beta}{128\beta} \cdot n$.
    Thus, by choosing the probability high enough and by using a union bound over the $T = \poly(n)$ many time steps, we can ensure that, with probability at least $1 - n^{-6}$, level $0$ does not become invalid at any time $ \in \set{1, 2, \dots, T}$.
    In particular, this yields
    \begin{equation}
    \Pr{\eCross_{0,\intco{t', \tau}}}
    \leq
    \Pr{\eCross_{0,\intco{t', T}}}
    \leq
    n^{-6}
    .
    \end{equation}

\item[Case~2:]
    For any level $\ell \in \set{1, 2, \dots, \ell^*}$ and large enough $n$ we get
    \begin{equation}
    \begin{aligned}
    \Pr{\eCross_{\ell,\intco{t', \tau}}}
    &\leq
    \Pr{\eCross^W_{\ell,\intco{t', \tau}}}
    \leq
    \Pr{\eCross^W_{\ell,\intco{t', T}}}
    \\&
    \overset{\mathllap{\text{Lem.~\ref{lem:biasedrw_crossprob}}}}{\leq}
    \frac{2^1 - 1}{2^{1 + \alpha_{\ell}/2} - 1}
    \leq
    \frac{2}{2^{1 + \alpha_{\ell}/2}}
    =
    \frac{1}{2^{\alpha_{\ell}/2}}
    \leq
    \frac{1}{2^{\alpha_{\ell^*}/2}}
    =
    n^{-6}
    .
    \end{aligned}
    \end{equation}

\item[Case~3:]
    For level $\ell = \ell^* + 1$ we get
    \begin{equation}
    \begin{aligned}
    \Pr{\eCross_{\ell^*,\intco{t', \tau}}}
    &\leq
    \Pr{\eCross^W_{\ell^*,\intco{t', \tau}}}
    \leq
    \Pr{\eCross^W_{\ell^*,\intco{t', T}}}
    \\&
    \overset{\mathllap{\text{Lem.~\ref{lem:biasedrw_crossprob}}}}{\leq}
    \frac{\sqrt{n}^1 - 1}{\sqrt{n}^{1 + \alpha_{\ell^*+1}/2} - 1}
    \leq
    \frac{\sqrt{n}}{\sqrt{n}^{1 + \alpha_{\ell^*+1}/2}}
    =
    n^{-6}
    .
    \end{aligned}
    \end{equation}
\end{description}
\medskip

By this case distinction, we know that the probability that any single level that just became critical at time $t' \in \intco{0, \tau}$ has probability at most $n^{-6}$ to become critical within $\intco{0, \tau}$.
Taking a union bound over the $\ell^* + 2 = \log\log(n) + \Theta(1)$ many levels and $T = n^4$ many time steps finishes the proof of \cref{lem:maintainvalidity}.
\end{proof}

    \section{Lower Discrepancy Bound when Deletions are Dominant.}%
\label{lower}

\Cref{thm:main2} gives a lower bound that shows that our result from \cref{thm:potential_bound} is tight in the sense that there are insertion probability sequences that result in a logarithmic discrepancy.
Our lower bound statement assumes that the potential $\Phi_{\alpha}$ (see \cref{sec:larstom}) is linear in $n$, which holds with high probability (see \cref{lem:adaption_sauerwald}).
Similar to the results of \cite{DBLP:journals/rsa/PeresTW15} the discrepancy is heavily dependent on the one-choice operations, ball deletions in our case.
A long sequence of ball deletions creates deep holes in the bins and the frequency of ball insertions must be high enough to re-balance the load.
\begin{restatable}{theorem}{thmlower}\label{thm:main2}
Fix a step $t$ and assume $\Phi_{\alpha}(t) = O(n)$ for some constant $c$. Let $T= t+\frac{n\ln(n)}{2+4\epsilon}$ and assume that for every $\tau$ with $t\le \tau\le T$ we have $\beta(\tau) = 1/2-\epsilon$ with $\epsilon < 1/2$ being an arbitrary positive constant.
Then the expected number of bins at the end of step $T$ with load at least $\floor{m(T)/n} + \ln(n)/2$ is $\Omega\left(\sqrt{n}\right)$.
\end{restatable}
We note that the above result also holds (with basically the same proof) for varying insertion probabilities as long as fewer than $(T-t)\cdot (1/2-\epsilon)$ balls are inserted during the time interval $\intcc{t, T}$.

The proof idea for \cref{thm:main2} is as follows:
We first show that, as long as in step $t$ we have $\Phi(t) \leq c \cdot n$ for some constant $c$ (which holds with high probability), there exists a constant fraction of bins with a load which is at least $m(t)/n$ (\cref{obs:constant_fraction_above_average}).
If now during the next $O( n\log n)$ steps much more balls are deleted than inserted, the average load will go down and one of the bins with load at least $m(t)/n$ will receive no ball deletions at all.
This bin will have a discrepancy of $O(\log n)$. This holds even under the assumption that the bin does not receive any additional ball.

\begin{lemma}\label{obs:constant_fraction_above_average}
Assume $\Phi \leq c \cdot n$ for some constant $c$ and
let $r = r(n)$ denote the fraction of bins with a load of at least $\lfloor m(t)/n \rfloor$. Then $r = \Theta(1)$, i.e., a constant fraction of bins has load at least $\lfloor \frac{m(t)}{n} \rfloor$.
\end{lemma}
\begin{proof}
Let $B_1$ denote the set of bins with a load of at least $\lfloor m(t)/n \rfloor$ and $B_2$ the set of bins with fewer than $\lfloor m(t)/n \rfloor$ balls.
Since $r$ is defined as the fraction of bins in $B_1$, we have $|B_1| = r \cdot n$ and $|B_2| = (1-r) \cdot n$.
This implies that there are at least $(1-r) \cdot n$ holes below the level $\lfloor \frac{m(t)}{n} \rfloor$ and, in turn, at least $(1-r) \cdot n$ balls above the average load of $m(t)/n$.

Let $m(t)/n + \ell$ be defined as the average load of the bins in $B_1$.
Note that $\ell\ge 0$ because $m(t)/n$ is the average load and the bins in $B_1$ have a higher load than the ones in $B_2$.
From the definition, it follows that $$\ell \geq \frac{n \cdot (1-r)}{r \cdot n} = \frac{1-r}{r}.$$
Since the potential function $\Phi$ is minimal when all bins in $B_1$ have load $\ell$, we obtain
$$\sum_{i=1}^{r \cdot n} e^{\alpha \cdot \frac{1-r}{r}} \leq \Phi \leq c \cdot n$$
which implies $r \cdot e^{\alpha \cdot \frac{1-r}{r}} \leq c$.
If we view $r=r(n)$ as function in $n$ and assume $r(n)=o(1)$ then in
\[
r(n)\cdot e^{\alpha \left(\frac{1}{r(n)}-1\right)} = r(n)\cdot f(n)\quad\mbox{with}\quad f(n)=e^{\alpha \left(\frac{1}{r(n)}-1\right)}
\]
the exponent is a function growing as $1/r(n)=\omega(1)$.
This implies that regardless of the size of $\alpha$ (as long as it is positive) the $f(n)$ grows exponentially quicker than the $r(n)=o(1)$ decays, and so the whole $r(n)\cdot f(n)$ is in $\omega(1)$ and in particular not upper-bounded by the constant $c$.

Since $r(n) \leq 1 \in O(1)$, we conclude that $r(n)=\Theta(1)$.
\end{proof}

With this, we can prove \cref{thm:main2}:
\begin{proof}[Proof of \cref{thm:main2}]
First of all, a simple application of Chernoff bounds proves that, w.h.p., $\floor{m(t)/n} \ge \floor{m(T)/n} + \ln(n)/2$.
Fix a bin $i$ and define $\ell=T-t=n\ln(n)/(2+4\epsilon)$.
Then for large $n$, the probability that during all steps $\tau$ with $t\le \tau \le T$ no ball is deleted from bin $i$ can be bounded by (using the inequality $1 - x \geq e^{-2x}$ for $x \in \intcc{0, 0.795}$)
\begin{equation}
\left(1-\frac{1}{n}\right)^{(1-(1/2-\epsilon))\ell}
=
\left(1-\frac{1}{n}\right)^{(1/2+\epsilon)\ell}
\geq
\exp\left(-(1+2\epsilon)\frac{\ell}{n}\right)
=
\exp\left(-\ln(n)/2\right)
=
n^{-1/2}.
\end{equation}
From \cref{obs:constant_fraction_above_average} we get that, at time $t$, a constant fraction of bins have load at least $\floor{m(t)/n}$.
And by the above calculation, the expected number of those bins from which no ball is deleted is $\Omega(\sqrt{n})$.
Thus, at time $T$, in expectation there are $\Omega(\sqrt{n})$ bins that have a load of at least $\floor{m(t)/n} \geq \floor{m(T)/n} + \ln(n) / 2$.
\end{proof}

}{}
\end{document}